\documentclass[11pt,a4paper,reqno]{amsart}
\usepackage[hmargin={2.3cm,2.3cm},vmargin={2.5cm,2.5cm},centering]{geometry}
\usepackage{amssymb,amsmath,amsthm,amsfonts,amscd}

\usepackage{graphicx}
\usepackage[dvipsnames]{xcolor}
\usepackage{mathrsfs}
\usepackage[unicode,psdextra]{hyperref}
\usepackage[shadow,textwidth=18mm,textsize=tiny]{todonotes}
\usepackage{tikz-cd}
\usepackage[utf8]{inputenc}
\usepackage[english]{babel}
\usepackage{dsfont}
\usepackage{bm,bbm}

\usepackage{enumerate}
\usepackage{setspace}
\usepackage[font={scriptsize,bf}]{caption}
\usepackage{changepage}
\usepackage{booktabs}
\usepackage[sort&compress,capitalise,nameinlink]{cleveref}
\crefname{section}{\textsection}{\textsection}
\crefname{subsection}{\textsection}{\textsection}
\crefname{subsubsection}{\textsection}{\textsection}
\crefname{paragraph}{\textparagraph}{\textparagraph}
\crefname{thm}{Theorem}{Theorem}
\usepackage[cal=euler,scr=boondoxo,bb=ams,calscaled=1.07,scrscaled=1.07]{mathalfa}
\makeatletter
\def\mathalfa@frakscaled{s*[1]}
\makeatother
\DeclareFontFamily{U}{euf}{}%
\DeclareFontShape{U}{euf}{m}{n}{<-7>\mathalfa@frakscaled eufm5
  <7-9>\mathalfa@frakscaled eufm7
  <9->\mathalfa@frakscaled eufm10}{}%
\DeclareFontShape{U}{euf}{b}{n}{<-7>\mathalfa@frakscaled eufb5
  <7-9>\mathalfa@frakscaled eufb7
  <9->\mathalfa@frakscaled eufb10}{}%
\DeclareMathAlphabet{\mathfrak}{U}{euf}{m}{n}

\renewcommand{\Im}{\mathrm{Im}}
\renewcommand{\Re}{\mathrm{Re}}
          \newtheorem{thm}{Theorem}[section]
          \newtheorem{proposition}[thm]{Proposition}
          \newtheorem{lemma}[thm]{Lemma}
          
          \newtheorem{definition}[thm]{Definition}
          
          \theoremstyle{definition}

          \newtheorem{remark}[thm]{Remark}

\usepackage{csquotes}

\makeatletter
\renewcommand*{\@textcolor}[3]{%
  \protect\leavevmode
  \begingroup
    \color#1{#2}#3%
  \endgroup
}
\makeatother

\newcommand{\xv}{\mathbf{x}}

\newcommand{\kv}{\mathbf{k}}

\newcommand{\dom}{\mathscr{D}}
\newcommand{\diff}{\mathrm{d}}
\newcommand{\eps}{\varepsilon}

\newcommand{\zv}{\mathbf{z}}

\newcommand{\beq}{\begin{equation}}
\newcommand{\eeq}{\end{equation}}

\newcommand{\OO}{\mathcal{O}}

\newcommand{\R}{\mathbb{R}}
\newcommand{\N}{\mathbb{N}}

\newcommand{\bdm}{\begin{displaymath}}
\newcommand{\edm}{\end{displaymath}}
\newcommand{\bdn}{\begin{eqnarray}}
\newcommand{\edn}{\end{eqnarray}}
\newcommand{\bay}{\begin{array}{c}}
\newcommand{\eay}{\end{array}}
\newcommand{\ben}{\begin{enumerate}}
\newcommand{\een}{\end{enumerate}}
\newcommand{\beqn}{\begin{eqnarray}}
\newcommand{\eeqn}{\end{eqnarray}}
\newcommand{\bml}[1]{\begin{multline} #1 \end{multline}}
\newcommand{\bmln}[1]{\begin{multline*} #1 \end{multline*}}

\newcommand{\lf}{\left}
\newcommand{\ri}{\right}

\newcommand{\tx}{\textstyle}

\newcommand{\bra}[1]{\lf\langle #1\ri|}

\newcommand{\meanlrlr}[3]{\lf\langle #1\lf|#2\ri|#3\ri\rangle}

\renewcommand{\leq}{\leqslant}
\renewcommand{\geq}{\geqslant}

\numberwithin{equation}{section}

\begin{document}

\newenvironment{sistema}%
{\left\lbrace\begin{array}{@{}l@{}}}%
    {\end{array}\right.}  \bibliographystyle{abbrvnat} \title{Microscopic derivation of time-dependent point interactions}

\author[R.\ Carlone]{Raffaele Carlone}

\address{Dipartimento di Matematica e Applicazioni ``Renato Caccioppoli'', Università degli Studi di
  Napoli ``Federico II'', MSA, via Cinthia, I-80126, Napoli, Italy.}

\email{raffaele.carlone@unina.it}

\urladdr{}

\author[M.\ Correggi]{Michele Correggi}

\address{Dipartimento di Matematica, Politecnico di Milano, P.zza Leonardo da Vinci, 32, 20133, Milano, Italy.}

\email{michele.correggi@gmail.com}

\urladdr{https://sites.google.com/view/michele-correggi}

\author[M.\ Falconi]{Marco Falconi}

\address{Dipartimento di Matematica e Fisica, Universit\`{a} di Roma Tre, L.go S. Leonardo Murialdo, 1/C, 00146 Roma, Italy.}

\email{mfalconi@mat.uniroma3.it}

\urladdr{http://ricerca.mat.uniroma3.it/users/mfalconi/}

\author[M.\ Olivieri]{Marco Olivieri}

\address{Fakult\"{a}t f\"{u}r Mathematik, Karlsruher Institut f\"{u}r Technologie, D-76128, Karlsruhe, Germany.}

\email{marco.olivieri@kit.edu}
\urladdr{}

\date{\today}

\begin{abstract}
  We study the dynamics of the three-dimensional polaron -- a quantum particle coupled to bosonic fields -- in the quasi-classical
  regime. In this case the fields are very intense and the corresponding degrees of freedom can be
  treated semiclassically. We prove that in such a regime the effective dynamics for the quantum particles
  is approximated by the one generated by a time-dependent point interaction, {\it i.e.}, a singular
  time-dependent perturbation of the Laplacian supported in a point. As a by-product, we also show that the
  unitary dynamics of a time-dependent point interaction can be approximated in strong operator topology
  by the one generated by time-dependent Schr\"{o}dinger operators with suitably rescaled regular
  potentials.
\end{abstract}

\maketitle

\section{Introduction and Main Results}
\label{sec:main}

Point interactions, also called zero-range interactions or Fermi pseudo-potentials have been widely studied in mathematical
physics as solvable models for realistic quantum systems (see
\cite{albeverio2005sqm} for an extensive review of the topic). By formally replacing a typically complicated
interaction potential with a sum of singular distributions (Dirac deltas) supported at isolated points or
on curves or surfaces, one aims at summing up all the features of the interaction in a minimal number of
physical parameters (\emph{e.g.}, scattering length, effective range, \dots). The models obtained by this
approximation are almost exactly solvable, while the salient physical features of the original systems
should be retained in the procedure. Often, however, there is no conclusive evidence of the latter
assertion: despite being known that point interactions can be approximated by suitable rescaled
potentials, it is not clear whether they can be derived from realistic models, in suitable physical
regimes.

In this paper we would like to derive a class of zero-range models as the effective description of
physical systems, originating from a well-defined approximation that we
call \emph{quasi-classical limit} (see \cite{correggi2017ahp,correggi2017arxiv,CFO} and references therein). This should clarify the importance of zero-range models in
mathematical and theoretical investigations. As a by-product, we also prove that there exist lattice field
quantum states in which a polaron is completely ionized (this is discussed in more detail below).

The class of zero-range models considered is the so-called \emph{time-dependent point
  interactions}: solvable models with singular potentials whose ``strength'' may change in time. They are
typically useful to investigate the ionization of a bound state by the action of a time-dependent
localized interaction. In the zero-range approximation, one studies the formal Hamiltonian
\begin{equation}
  \label{eq:formal}
  - \Delta + \mbox{``}\mu(t) \delta(\mathbf{x})\mbox{''}
\end{equation}
in $ L^2(\mathbb{R}^d) $, and the system is assumed to be in a bound state at initial time, \emph{e.g.}, in the
ground state, and the asymptotic probability of ionization is computed \cite{coretal2005,cordel2005}
(see also \cite{coscosleb2018} and \cite{cordel2004} for the one-dimensional version of the same
model and other time-dependent singular perturbations, respectively). As already discussed, the physical relevance of such a minimal model is still unclear and it is one of the main goals of the present investigation.

We prove that time-dependent point interactions can indeed be derived from the microscopic dynamics of a
quantum particle interacting with bosonic scalar quantum fields, in suitable configurations
in which the fields are very intense, and the average number of carriers is much larger
than one. In particular, we consider the coupling of the particle with two distinct species of
  force-carrying fields. The presence of several force fields is not uncommon in physical systems, both in
  condensed matter and at high energies. The model we are considering, see \eqref{eq:Heps} below, is one
  describing two species of phonons interacting with a quantum particle (\emph{e.g.}, an electron or a
  molecule) in a lattice. Although the original Fr\"{o}hlich polaron contains only one species of phonons, the optic ones (see \cite{frohlich1937prslA}), there are physical systems for which the presence of two species is meaningful. The
  typical example of a model in which quantum particles are coupled with both acoustic and optic phonons
  is that of a compound ionic crystal (see, \emph{e.g.}, \cite[Chpt. 4]{kittel}), {\it i.e.}, with at least two species of atoms per unit cell. Let us remark that it is the
  presence of \emph{two different} species of phonons, with different scaling properties and dispersion
  relations, that produces the ionization of the particle in our model. In mathematical terms, it is the interplay between the two species that yields a
  \emph{time-dependent} point interaction, instead of a time-independent one (see also \cref{rem:1}).

As mentioned above, we consider configurations in which the fields are very intense, due to the
  presence of a very large number of force carriers. Here, the reference scale is the number of non-relativistic particles in the system, {\it i.e.}, just 1 in our setting, which is also of the same order as the
commutator between creation and annihilation operators of the two fields, set equal
to one in the considered units. Equivalently, it is more convenient to set
\begin{equation}
  \label{eq:commutator}
  \left[ a_{\varepsilon}(\mathbf{x}), a^{\dagger}_{\varepsilon}(\mathbf{y}) \right]=\left[ b_{\varepsilon}(\mathbf{x}), b^{\dagger}_{\varepsilon}(\mathbf{y}) \right] = \varepsilon\, \delta(\mathbf{x}-\mathbf{y})\; ,
\end{equation}
and to take the limit
\begin{equation}
  \label{eq:qc limit}
  \varepsilon \to  0\; ,
\end{equation}
where the parameter $ \eps $ has the physical meaning of the inverse of the average number of field excitations.
Here $ a(\mathbf{x}), b(\mathbf{x}) $ and $ a^{\dagger}(\mathbf{x}),b^{\dagger}(\mathbf{x}) $, $ \mathbf{x} \in \mathbb{R}^3 $,
are the usual operator-valued distributions associated to the annihilation and creation of field
carriers, for the two species of phonons. The limit $\varepsilon\to 0$ is precisely the aforementioned
quasi-classical limit. In the case of the Fr\"{o}hlich's polaron, in which the particle is coupled with
optic phonons only, such a limit is known to be equivalent, at least in the stationary picture, to the
so-called {\it strong coupling} regime, {\it i.e.}, of a very intense coupling between the particles and
the Bose field. The strong coupling dynamics is however slightly different from the one considered here because
the field is frozen to leading order.

Polaron models were originally introduced in
\cite{frohlich1937prslA} to describe the interaction between one or more electrons with a crystal of
nuclei, vibrating around the rest positions on the lattice crystal, but, more recently, have been widely used in solid and condensed matter physics,
notably also as effective models to describe the behavior of an impurity in a Bose-Einstein condensate
(see, {\it e.g.}, \cite{GD} for a review of the topic). These latter models can be concretely
realized in the lab \cite{Hohmann}, by immersing a $ \mathrm{Cs} $ impurity in a $ \mathrm{Rb} $
condensate. Typically, the polaron is studied in three dimensions, but lower dimensional models may also
be of a certain interest from the physical point of view \cite{GAD}. For instance, one-dimensional
tight-binding models are known to well-approximate the motion of electrons in organic semiconductors
\cite{DeFilippis}. Furthermore, in these investigations, the phonon degrees of freedom are typically
treated classically, as in the ideal quasi-classical limit described above.

When the limit $ \varepsilon \to 0 $ is taken, the bosonic fields become classical fields, whose
dynamics \emph{a priori }would depend on both its initial configuration and the coupling with the
particle. We are going to scale the coupling in such a way that the classical field either evolves
freely or is frozen in the limit, while the effective particle dynamics is generated by an effective
time-dependent interaction, which is time-periodic and point-like at the origin, \emph{i.e.}, it is the
rigorous counterpart of~\eqref{eq:formal}. Physically speaking, the quasi-classical limit we are
considering is such that there is no back-reaction of the particle on the classical fields.

Mathematically, the dimensions lower than three are easier to deal with, therefore we prefer to focus on a
three-dimensional model. Our techniques could however be adapted to study one- and two-dimensional systems.
Notice that, since we derive a time-dependent point interaction from the polaron Hamiltonian and the former model shows complete asymptotic ionization \cite[Thm. 4.4]{coretal2005}, we also prove that there
exist quantum field configurations that yield complete ionization of the polaron (compare with \cite{Toyozawa}). Let us stress that our
model is at zero-temperature, and the ionization is due to the fields' configuration only; a discussion of
temperature-induced ionization for polaron systems can also be found in the physical literature (see,
\emph{e.g.}, \cite{hawton1979pla}).

\medskip

In the next \cref{sec:micro,sec:effective}, we describe our setting in more detail, introducing the
microscopic system and the effective point interaction model, respectively. The convergence of the
dynamics is discussed in \cref{sec:results}, where we also state a technical result about the
approximation of the dynamics generated by a time-dependent point interaction by means of the dynamics
generated by the corresponding approximating regular potentials.

\medskip

Let us conclude this section by fixing some basic notations. We use the convention of denoting by
calligraphic letters all the quantities referring to the effective model, \emph{e.g.}, $
\mathcal{U}_{\mathrm{eff}} $ and $ \mathcal{H}_{\mathrm{eff}} $ denote the effective dynamics and
generator, respectively, while regular letters (\emph{e.g.}, $ U_{\varepsilon} $ and $ H_{\varepsilon} $) are attached to the
microscopic counterparts, \emph{i.e.}\ to the operators acting on the full particle-field Hilbert
space. Bold letters denotes vectors ($\mathbf{x}$), while italic roman letters are used for scalar
quantities. When there is no risk of confusion, we use the corresponding italic letter ($x$) to denote the
modulus of a vector ($\mathbf{x}$). Concerning operators and quadratic forms, we denote by $ \mathscr{D}(A) $ and $ \mathscr{D}[A] $ the operator and form domains, respectively.

\subsection{Microscopic model}
\label{sec:micro}

As anticipated, we want to investigate the quasi-classical limit of a system composed of a quantum
spinless\footnote{The spin can be easily added to the picture, however since it would make the notations
  more cumbersome, we avoid it. Similarly, we can include more quantum particles in the model and possibly
  some interaction or a trapping potential, but we consider the simplest model for the sake of clarity.}
particle interacting with two bosonic fields. In order to discuss a model that is as
realistic as possible, we consider the interaction with acoustic and optic phonon lattice
fields. Other models could however be considered as well.

The Hilbert space of microscopic states is thus $\mathscr{H} := L^2(\mathbb{R}^3) \otimes \Gamma_{\mathrm{sym}}(\mathfrak{H})\otimes
  \Gamma_{\mathrm{sym}}(\mathfrak{H})$, where the single excitation (phonon) space is $\mathfrak{H}= L^2(\mathbb{R}^3)$, and $ \Gamma_{\mathrm{sym}} $
stands for the symmetric Fock space.  The Hamiltonian of the full system reads
\begin{multline}	
	\label{eq:Heps}
     	H_{\varepsilon}=H_0+H_I=-\Delta\otimes 1\otimes 1+ 1\otimes \mathrm{d}\Gamma_{\varepsilon}^{(a)}(\omega)\otimes 1 + \frac{\kappa}{\varepsilon} 1\otimes 		1\otimes \mathrm{d}\Gamma^{(b)}_{\varepsilon}(1)	\\
     +a_{\varepsilon}\lf(\lambda^{(a)}_{\mathbf{x}}\ri)+a^{\dagger}_{\varepsilon}\lf(\lambda^{(a)}_{\mathbf{x}}\ri)+\,b_{\varepsilon}\lf(\lambda^{(b)}_{\mathbf{x}}\ri)+b^{\dagger}_{\varepsilon}\lf(\lambda^{(b)}_{\mathbf{x}}\ri)
\end{multline}
where $ 1 $ stands for the identity operator on either $ L^2(\mathbb{R}^3) $ or $\Gamma_{\mathrm{sym}}(\mathfrak{H}) $, and 
\beq
	H_0=-\Delta\otimes
	1\otimes 1+1\otimes \mathrm{d}\Gamma_{\varepsilon}^{(a)}(\omega)\otimes 1 + \frac{\kappa}{\varepsilon}\bigl(1\otimes 1\otimes \mathrm{d}\Gamma^{(b)}_{\varepsilon}(1)\bigr).
\eeq
The creation and annihilation operators $a_{\varepsilon},a_{\varepsilon}^{\dagger}$ and $b_{\varepsilon},b_{\varepsilon}^{\dagger}$ refer to the acoustic
  and optic phonons respectively, and $\kappa>0$ is a frequency parameter for the optic phonons. Although more general choices are possible, we assume that the
  dispersion relation $\omega$ of the acoustic phonons is simply
  \begin{equation}
    \label{eq:10}
    	\omega(\kv) = k.
  \end{equation}
  In the Hamiltonian above, we have used the
  $\varepsilon$-dependent representation of the canonical commutation relations for the creation and annihilation
  operators $ a_{\varepsilon}, b_{\varepsilon} $ and $ a_{\varepsilon}^{\dagger},b^{\dagger}_{\varepsilon} $, {\it i.e.},
\begin{equation}
  \label{eq:ccr}
  \left[ a_{\varepsilon}(\xi), a^{\dagger}_{\varepsilon}(\eta) \right] =\left[ b_{\varepsilon}(\xi), b^{\dagger}_{\varepsilon}(\eta) \right]= \varepsilon \left\langle \xi|\eta \right\rangle_{\mathfrak{H}}\; ;
\end{equation}
however the two fields have kinetic terms that scale differently in the quasi-classical limit. The
  acoustic phonons yield a \emph{time-independent} effective potential, while the optic phonons yield a
  \emph{time-dependent} one. The form factor for the acoustic phonons
  $\lambda_{\mathbf{x}}^{(a)}(\mathbf{k})=e^{i \mathbf{k}\cdot \mathbf{x}}\lambda_0(\mathbf{k})$ is such that $\lambda_0$ is
  invertible for almost all $\mathbf{k}\in \mathbb{R}^3$, and 
\begin{equation}
  \label{eq:12}
  \lambda_0(\mathbf{k})\text{ , }k^{-1/2}\lambda_0(\mathbf{k})\in \mathfrak{H}\; .
\end{equation}
Furthermore, we assume that $ \lambda_0 $ is polynomially decaying at large $ k $, or, more precisely,
  \begin{equation}
    \label{eq:11}
    \frac{1}{\lf| \lambda_0(\mathbf{k})\ri|}  \leq C (k^2+1)^{M}, \quad \text{ for some } M\geq 0\; .
  \end{equation}
	On the other hand, the form factor for the optic phonons $\lambda^{(b)}_\mathbf{x}$ has the explicit form
\begin{equation}
  	\lambda^{(b)}_\mathbf{x}(\mathbf{k}) = \frac{e^{i \mathbf{k} \cdot \mathbf{x}}}{k}\; .
\end{equation}

The commutation relations \eqref{eq:ccr} can be thought of as a quasi-classical rescaling of the usual
relations; such rescaling is convenient to investigate the regime $ \left\langle a^{\dagger} a \right\rangle\sim \left\langle b^{\dagger} b
\right\rangle \gg [a, a^{\dagger}]= [b, b^{\dagger}]= 1 $, where $a^{\sharp},b^{\sharp}$ are the usual creation and annihilation
operators. In fact, one can set $ a^{\sharp}_{\varepsilon} : = \sqrt{\varepsilon} a^{\sharp} $, $ b^{\sharp}_{\varepsilon} : = \sqrt{\varepsilon} b^{\sharp} $,
from which \eqref{eq:ccr} follows. Moreover, given any self-adjoint one-particle operator $h$ on
$\mathfrak{H}$,
\begin{equation}
	\label{eq:dgamma}
	\mathrm{d}\Gamma^{(a)}_{\varepsilon}(h)=\varepsilon \mathrm{d}\Gamma^{(a)}(h)\;, \qquad		\mathrm{d}\Gamma^{(b)}_{\varepsilon}(h)=\varepsilon \mathrm{d}\Gamma^{(b)}(h)\; ,
\end{equation}
where the left hand side is written w.r.t.\ $a^{\sharp}_{\varepsilon}, b^{\sharp}_{\varepsilon}$ and the latter w.r.t.\ $a^{\sharp},
b^{\sharp}$. The parameter $\varepsilon$ is the quasi-classical parameter that describes the energy scale of the
macroscopic phonon field, which is of order $ \mathcal{O}(\varepsilon^{-1}) $, and is assumed to be small,
\emph{i.e.}, $ \varepsilon \ll 1 $. This corresponds to high field energies, due to the presence of a large number
of excitations.

It is worth noting that if we get rid of the optic phonons, our quasi-classical limit is
equivalent, up to a suitable rescaling of time and lengths, to the {\it
  strong-coupling regime}. We refer to
\cite{griesemer2013jpa,frank2014lmp,frank2017apde,griesemer2016arxiv2,lieb1997cmp} for further details on
the strongly coupled polaron. 

Since for any $ \mathbf{x} \in \mathbb{R}^3 $, $\lambda^{(b)}_\mathbf{x} \notin \mathfrak{H}$, $H_{\varepsilon}$ can be written explicitly
as the above sum only as a quadratic form, acting on the form domain of the non-interacting part $
\mathscr{D}[H_0] $ (see, {\it e.g.},
\cite{lieb1997cmp,frank2014lmp,falconi2015mpag,griesemer2015arxiv}). The form $ \left\langle \: \cdot
  \:\right|H_{\varepsilon}\left| \: \cdot \: \right\rangle_{\mathscr{H}}$ is closed and bounded from below on $\mathscr{D}[H_0]$,
and therefore $H_{\varepsilon}$ can be defined as a self-adjoint operator on a suitable dense domain $
\mathscr{D}(H_{\varepsilon}) \subset \mathscr{D}[H_0]$. Moreover, $H^2(\mathbb{R}^3) \otimes \mathscr{D}(\mathrm{d}\Gamma^{(a)}_{\varepsilon}(\omega)^{1/2})\otimes
\mathscr{D}(\mathrm{d}\Gamma^{(b)}_{\varepsilon}(1)^{1/2}) $ is dense in $\mathscr{H}$ and contained in the form domain
$\mathscr{D}[H_0]$.

\subsection{Effective dynamics}
\label{sec:effective}

As anticipated above, we want to derive an effective dynamics generated by the formal operator
\eqref{eq:formal}, \emph{i.e.}, $ - \Delta + \mbox{``}\mu(t) \delta(\mathbf{x})\mbox{''} $. In two or three dimensions
the rigorous definition of the self-adjoint counterpart of this formal expression is not straightforward:
one can not simply consider the corresponding energy form and investigate his closedness, as it is done in
one dimension. The typical way to address this question (see, \emph{e.g.}, \cite{albeverio2005sqm}) is
to consider the operator $ - \Delta $ restricted to functions vanishing at the origin and classify its
self-adjoint extensions, which form a one-parameter family of operators $ \lf\{\mathcal{H}_{\beta} \ri\}_{\beta \in \mathbb{R}} $,
given by
\begin{equation}
  \begin{array}{ll}
    \mathcal{H}_{\beta} \psi &=   - \Delta \phi,		 \\
    \mathscr{D}(\mathcal{H}_{\beta}) &=    \left\{ \psi \in L^2(\mathbb{R}^3) \: \Big|  \: \psi = \phi + \displaystyle\frac{q}{4\pi x}, \phi \in H^2_{\mathrm{loc}}(\mathbb{R}^3), \Delta \phi \in L^2(\mathbb{R}^3), q \in \mathbb{C}, \, \phi(0) = \beta q \right\}.
    \label{eq:eff ham}
  \end{array}
\end{equation}

The interaction is thus encoded into the boundary condition $ \phi(0) = \beta q $, which has to be satisfied by
any function in $ \mathscr{D}(\mathcal{H}_{\beta}) $. The action of $ \mathcal{H}_{\beta} $ on the other hand
coincides with the one of $ - \Delta $, although on the regular part of the wave function $ \phi $. In
particular, if $ \phi(0) = 0 $, it is immediate to verify that $ \mathcal{H}_{\beta} \phi = - \Delta \phi $ and, in this
respect, $ \mathcal{H}_{\beta} $ defines a self-adjoint realization of the formal expression $ - \Delta +
\mbox{``}\mu \delta(\mathbf{x})\mbox{''} $. However, it is important to remark that the meaning of the parameter
$ \beta $ in \eqref{eq:eff ham} is not the strength of the interaction but it is rather proportional to the inverse of the
scattering length, so that for instance the free operator $ - \Delta $ corresponds to $ \beta = + \infty $.

As already discussed, our main goal is thus to prove a rigorous derivation of the effective particle
dynamics generated by the time-dependent operator $ \mathcal{H}_{\beta(t)} $, with $ \beta(t) $ a periodic
function. The existence of such a dynamics has already been studied in the literature \cite{sayyaf1984}
(see also \cite{carcorfig2017,carfiote2017,delfigtet2000,posilicano2007,SY2,Ya} for further details and similar results) and it is known that,
under minimal regularity assumptions on $ \beta(t) $, \emph{i.e.}\ $ \beta(t) \in L^{\infty}_{\mathrm{loc}}(\mathbb{R}) $, there
exists a two-parameter unitary group $ \mathcal{U}_{\mathrm{eff}}(t,s) $, $ t,s \in \mathbb{R} $, which is generated
by $ \mathcal{H}_{\beta(t)} $: for any $ s \in \mathbb{R} $ and any $ \psi \in \mathscr{D}(\mathcal{H}_{\beta(s)}) $,
\begin{equation}
  \label{eq:generator}
  \lim_{h \to 0} \frac{i \left( \mathcal{U}_{\mathrm{eff}}(s+h,s) - 1 \right) \psi}{h} = \mathcal{H}_{\beta(s)} \psi\; .
\end{equation}

In fact, the time-evolution generated by $ \mathcal{H}_{\beta(t)} $ can be explicitly characterized: for any $
\psi_s \in H^2_0(\mathbb{R}^3 \setminus \{0\}) $,
\begin{equation}
  \label{eq:ansatz}
  \psi_t(\mathbf{x}) : = \left( \mathcal{U}_{\mathrm{eff}}(t,s) \psi_s \right)(\mathbf{x}) = \left( U_0(t-s) \psi_s \right)(\mathbf{x}) + i \int_s^t \mathrm{d} \tau \: U_0(t - \tau; \mathbf{x}) q(\tau),
\end{equation}
where we have denoted by $ U_0(t) $ the unitary time-evolution generated by $ \mathcal{H}_0 = - \Delta $, and
by $ U_0(t; \mathbf{x}) $ the corresponding integral kernel, which reads
\begin{equation}
  U_0(t; \mathbf{x}) : = \frac{1}{(4\pi i t)^{3/2}} \exp \left\{ \frac{i x^2}{4t} \right\}\, .
\end{equation}
The {\it charge} $ q(t) \in \mathbb{C} $ solves the Volterra-type integral equation
\begin{equation}
  \label{eq:charge eq}
  q(t) + 4 \sqrt{\pi i} \int_s^t \mathrm{d} \tau \: \frac{\beta(\tau) q(\tau)}{\sqrt{t - \tau}} = 4 \sqrt{\pi i} \int_s^t \mathrm{d} \tau \: \frac{1}{\sqrt{t - \tau}} \left(U_0(\tau - s) \psi_s \right)(0),
\end{equation}
and, in fact, is the unique solution of such an equation in the space of continuous functions. Notice that
it is easy to verify heuristically that \eqref{eq:ansatz} solves the time-dependent Schr\"{o}dinger
equation, since
\begin{displaymath}
  i \partial_t \psi_t = - \Delta \left( \psi_t - \frac{q(t)}{4\pi x} \right),
\end{displaymath}
while it is much harder to derive the charge equation \eqref{eq:charge eq}, which is related to the
boundary condition in \eqref{eq:eff ham}. The evolution $ \mathcal{U}_{\mathrm{eff}} $ can be expressed as
in \eqref{eq:ansatz} for any $ \psi_s $ vanishing around the origin, but it can be extended to any state in $
L^2(\mathbb{R}^3) $ by density.

\subsection{Main results}
\label{sec:results}

We now describe in more detail the quasi-classical regime. Let $ \mathcal{B} \in \mathscr{B}\bigl(L^2
(\mathbb{R}^3)\bigr)$ be a bounded\footnote{The analysis can be easily extended to suitable unbounded
  observables, but we restrict the discussion to bounded operators for the sake of simplicity.}.  observable associated to the particle. The Heisenberg evolution of $ \mathcal{B} $ in the
microscopic dynamics is given by
\begin{equation}
  \label{eq:heisenberg}
  B_{\varepsilon}(t,s)=e^{iH_{\varepsilon}(t-s)} \mathcal{B} \otimes 1 \,e^{-iH_{\varepsilon}(t-s)}.
\end{equation}
Since the microscopic system describes an interaction between the particle and the phonon fields, for almost all $t \neq s \in \mathbb{R} $, $ B_{\varepsilon}(t,s)\Bigr\rvert_{\Gamma_{\mathrm{sym}}(\mathfrak{H}) \otimes \Gamma_{\mathrm{sym}}(\mathfrak{H})}\neq 1$, thus acting non-trivially on
the field. However, since we are interested on a description of the subsystem consisting of the particle
alone, let us take the partial expectation of $ B_{\varepsilon}(t,s)$ with respect to some initial state of the fields $\Psi_{\varepsilon}\in
\Gamma_{\mathrm{sym}}(\mathfrak{H})\otimes \Gamma_{\mathrm{sym}}(\mathfrak{H})$ at time $ s $ (to be fixed later): 
\bml{
  \label{eq:ev:start}
  \mathcal{B}_{\varepsilon}(t,s):= \bigl\langle \Psi_{\varepsilon}\bigl\lvert B_{\varepsilon}(t,s)\bigr\rvert \Psi_{\varepsilon}\bigr\rangle_{\Gamma_{\mathrm{sym}}(\mathfrak{H})\otimes \Gamma_{\mathrm{sym}}(\mathfrak{H})} \\
  = \bigl\langle e^{-iH_{\varepsilon}(t-s)} \Psi_{\varepsilon}\bigl\lvert\mathcal{B} \otimes 1 \bigr\rvert e^{-iH_{\varepsilon}(t-s)} \Psi_{\varepsilon} \bigr\rangle_{\Gamma_{\mathrm{sym}}(\mathfrak{H})\otimes \Gamma_{\mathrm{sym}}(\mathfrak{H})} \in \mathscr{B} \bigl(L^2(\mathbb{R}^3) \bigr).
}
Our main goal is thus to prove that in the quasi-classical limit $\varepsilon\to 0$, $\mathcal{B}_{\varepsilon}(t,s)$ converges
to $\mathcal{B}(t,s)$ in some topology (that turns out to be at most the strong operator topology), the
latter being defined as
\begin{equation}
  \label{eq:ev:obs}
  \mathcal{B}(t,s):=\mathcal{U}_{\mathrm{eff}}^{\dagger}(t,s) \, \mathcal{B} \; \mathcal{U}_{\mathrm{eff}}(t,s)\; ,
\end{equation}
where the two-parameter unitary group $ \mathcal{U}_{\mathrm{eff}}(t,s) $ is defined in 
\cref{sec:effective}. In other words, the effective quasi-classical evolution of the particle is, whenever
the field is in a suitable microscopic state described by $\Psi_{\varepsilon}$, generated by the time-dependent point
interaction Hamiltonian $ H_{\beta(t)} $. The expression \eqref{eq:ev:obs} does not hold true for any choice
of the field state $\Psi_{\varepsilon}$, but only for a restricted class of vectors; the relevant example on which we
focus our attention is a time-dependent coherent state of the form
\begin{equation}
  \label{eq:coherent}
  \Xi_{\varepsilon,s}:= W^{(a)}_{\varepsilon}\left(\tfrac{\alpha_{\varepsilon}}{i\varepsilon}\right) W^{(b)}_{\varepsilon}\left(\tfrac{\beta_{\varepsilon}}{i\varepsilon}\right)\Omega\; ,
\end{equation}
where
\begin{equation}
  W^{(a)}_{\varepsilon}(\xi) = e^{i(a_{\varepsilon}(\xi) + a_{\varepsilon}^{\dagger}(\xi))}\;,	\qquad W^{(b)}_{\varepsilon}(\xi) = e^{i(b_{\varepsilon}(\xi) + b_{\varepsilon}^{\dagger}(\xi))}
\end{equation}
are the Weyl operators for the two fields and $\Omega$ is the $\varepsilon$-independent Fock vacuum vector of
$\Gamma_{\mathrm{sym}}(\mathfrak{H})\otimes \Gamma_{\mathrm{sym}}(\mathfrak{H})$. Our proof takes advantage of the convenient coherent structure of
$\Xi_{\varepsilon,s}$. The functions $ \alpha_{\varepsilon},\beta_{\varepsilon} \in \mathfrak{H} $ are chosen of the following
form\footnote{When $\kappa'=0$, the optic phonons' subsystem decouples from the other subsystems, and
    therefore it does not play a role in the analysis. The choice of $\beta_{\varepsilon}$ given in \eqref{eq:alpha} is
    thus to be intended only for $\kappa'>0$.}:
\begin{equation}
  \label{eq:alpha}
	\alpha_{\varepsilon} (\mathbf{k}) = c_{\varepsilon}\sigma_\varepsilon \lambda_0^{-1}(\mathbf{k})\widehat{\mathcal{W}}\left(\sigma_\varepsilon \mathbf{k} \right)\;, \qquad \beta_{\varepsilon} (\mathbf{k}) = c'_{\varepsilon} \sigma_\varepsilon\, k\,\widehat{\mathcal{W}}\left(\sigma_\varepsilon \mathbf{k} \right)\;,
\end{equation}
where $ \mathcal{W} \in C^{\infty}_0(\mathbb{R}^3) $ is a resonant potential in the
sense of \cref{defi:resonant}, and $ \widehat{\mathcal{W}} $ stands for its Fourier transform, $ c_{\varepsilon},c'_{\varepsilon} $
are $\varepsilon$-uniformly bounded constants
\begin{equation}
  c_{\varepsilon} = \tx\frac{1}{2} \gamma_a \sigma_{\varepsilon}+o(\sigma_{\varepsilon})\;,	\qquad c'_{\varepsilon} = \tx\frac{1}{2}\gamma_b \sigma_{\varepsilon}+o(\sigma_{\varepsilon})\; ,
\end{equation}
with $ \gamma_a, \gamma_b \in \R $ and $ \sigma_{\varepsilon} > 0 $ is such that
\begin{equation}
\varepsilon^{1/{j_{*}}} \ll \sigma_{\varepsilon} \ll 1\;,
\end{equation}
where $ j_* := 6+8M$ is inherited from \eqref{eq:11}.

We can now state our main result:

\begin{thm}[Effective dynamics]
  \label{thm:eff}
  \mbox{}	\\
  Let $ \mathcal{B} $ be a bounded operator on $ L^2(\mathbb{R}^3) $. Let also $ \mathcal{B}_{\varepsilon}(t,s) $ and $
  \mathcal{B}(t,s) $ be defined as in \eqref{eq:ev:start} and \eqref{eq:ev:obs}, respectively. If the field
  is in the coherent state \eqref{eq:coherent} at time $ s \in \mathbb{R} $, then, for any $ t \in \mathbb{R} $,
  \begin{equation}
    \label{eq:convergence}
    \mathcal{B}_{\varepsilon}(t,s) \xrightarrow[\varepsilon \to 0]{\mathrm{s}} \mathcal{B}(t,s) = \mathcal{U}_{\mathrm{eff}}^{\dagger}(t,s) \, \mathcal{B} \; \mathcal{U}_{\mathrm{eff}}(t,s),
  \end{equation}
  where $ \mathcal{U}_{\mathrm{eff}} $ is the dynamics generated by $ \mathcal{H}_{\beta(t)} $, with
  \begin{equation}
    \beta(t) = \gamma_a +\gamma_b \cos  \kappa(t-s).
  \end{equation}
\end{thm}

\begin{remark}[Time-dependence]
  \label{rem:1}
  \mbox{}        \\
  It is also possible to obtain time-independent point interactions in the limit, \emph{i.e.}, to
    prove the analogous of \cref{thm:eff} with $\beta(t)$ replaced by $\gamma_1$. This is achieved whenever the
    optic phonons do not play any role in the limit. Concretely, this can be done by choosing suitable
    initial configurations for the optical field $b^{\sharp}_{\varepsilon}$: possible choices are the vacuum, or suitable coherent, whose average number of excitations is subleading w.r.t. $ 1 $.
\end{remark}
\begin{remark}[Initial state]
  \label{rem:initial}
  \mbox{}	\\
  The choice of the initial state $\Xi_{\varepsilon,s}$ is crucial for the derivation of the effective dynamics. Let us
  stress that $\Xi_{\varepsilon,s}$ depends on $\varepsilon$ through the phonon fields (of order $\sqrt{\varepsilon}$), through the factor
  $\varepsilon^{-1}$ in the argument, and additionally through $\alpha_{\varepsilon},\beta_{\varepsilon} $. This latter $\varepsilon$-dependence, that is new
  compared to other situations in which coherent states are used to investigate quasi- and semi-classical
  limits (see,
  \emph{e.g.}, \cite{hepp1974cmp,ginibre1979cmpI,ginibre1979cmp2,ginibre2006ahp,rodnianski2009cmp,falconi2013jmp}),
  is chosen so that the effective potential becomes singular in the limit, due to the $\sigma$-scaling: in the
  language of semiclassical analysis, $\Xi_{\varepsilon,s}$ converges to an additive but not $\sigma$-additive cylindrical
  Wigner measure on the classical space of fields (see \cite{falconi2017ccm,falconi2017arxiv} for additional details on cylindrical Wigner
  measures). In particular, the results proven in \cite{CFO} do not apply to such a class of states, which call for an alternative approach.
\end{remark}
	
\begin{remark}[Field dynamics]
  \label{rem:field}
  \mbox{}	\\
  The effective dynamics $ \mathcal{U}_{\mathrm{eff}} $ of the particle is time-dependent in the
  quasi-classical regime, although the original microscopic dynamics $ U_{\varepsilon} $ generated by $ H_{\varepsilon} $ is
  time-independent. The reason is clearly the interaction of the particle with the optic
    phonons, which produces an entanglement of any initial product state, and gives rise, together
    with the particle-acoustic phonons interaction, to the non-trivial dynamics $
  \mathcal{U}_{\mathrm{eff}} $ for the particle in the limit $ \varepsilon \to 0 $. One may wonder however what is the
  effect on the fields themselves of such an interaction and, as we are going to show, such an
  effect actually disappears as $ \varepsilon \to 0 $, since in the chosen scaling the classical acoustic field is
    constant in time, while the optic field undergoes an evolution which is asymptotically free: if we
  denote by $ \alpha(\mathbf{k}; t), \beta(\mathbf{k};t) $ the classical counterparts of the field
  operators, they satisfy $ i \dot{\alpha}(t) = 0 , i \dot{\beta}(t)= \kappa \beta(t)$, or, equivalently,
  \begin{equation}
  	\label{eq: alpha evolution}
        \alpha(\mathbf{k};t)=\alpha(\mathbf{k})\;,	\qquad \beta(\mathbf{k}; t) = e^{-i\kappa(t-s)} \beta(\mathbf{k})\;,
  \end{equation}
 	where $ \alpha(\kv) $ and $ \beta(\kv) $ stands for the classical fields at initial time $ s $., {\it i.e.}, the classical counterparts of \eqref{eq:alpha}.
  In other words, there is no back-reaction of the particle on the fields in the chosen
  quasi-classical scaling.
\end{remark}
	
\begin{remark}[Heisenberg evolution]
  \label{rem:heisenberg}
  \mbox{}	\\
  The convergence proven in~\eqref{eq:convergence} implies by duality that the effective Schr\"{o}dinger
  dynamics on states in $ L^2(\mathbb{R}^3) $ is also given by the two-parameter unitary group $
  \mathcal{U}_{\mathrm{eff}}(t,s) $: in the limit $ \varepsilon \to 0 $, a state $ \psi \in L^2(\mathbb{R}^3) $ at time $ s $ is
  mapped to $ \mathcal{U}_{\mathrm{eff}}(t,s) \psi $ at time $ t $.
\end{remark}

An important step in the proof of the main Theorem above is the approximation of the effective dynamics $
\mathcal{U}_{\mathrm{eff}} $ by the one generated by Schr\"{o}dinger operators with suitable regular
potentials. It is indeed very well known {\cite[Sect. I.1.2]{albeverio2005sqm}} that a point interaction Hamiltonian can be obtained as the
strong resolvent limit of a sequence of operators with rescaled smooth potentials: let $ \mathcal{W} \in C^{\infty}_0(\mathbb{R}^3)
$ and let $ 0 < \sigma \ll 1 $; set
\begin{equation}\label{eq:vsigma}
  \mathcal{W}_{{\beta},\sigma}(\mathbf{x}) : = \nu(\sigma) {\mathcal{W}_{\sigma}(\xv)},	\qquad		{\mathcal{W}_{\sigma}(\xv) : = \frac{1}{\sigma^2} \mathcal{W}(\xv/\sigma),}
\end{equation}
{for some $ \nu \in C([0,1]) $ given in \eqref{eq:nu} below} and
\begin{equation}
  \label{eq:happ}
  \mathcal{K}_{\beta,\sigma} : = - \Delta + \mathcal{W}_{{\beta},\sigma}(\mathbf{x}),
\end{equation}
which is obviously self-adjoint on $H^2(\mathbb{R}^3) $. In order to generate a point interaction, it is well-known {(see again \cite[Sect. I.1.2]{albeverio2005sqm})} that a zero-energy resonance must be present. Therefore, we formulate the following

	\begin{definition}[Resonant potential]
		\label{defi:resonant}
		\mbox{}	\\
		Let $ \mathcal{W} \in C_0^{\infty}(\mathbb{R}^3) $. We say that $ \mathcal{W} $ is resonant, if 
		\begin{itemize}
			\item $ -1 $ is a simple eigenvalue of $ \mathrm{sgn}(\mathcal{W}) |\mathcal{W}|^{1/2} {(-\Delta)}^{-1} |\mathcal{W}|^{1/2} $ with eigenvector $ \phi \in L^2(\mathbb{R}^3) $;
			\item $ \left.\left\langle |\mathcal{W}|^{1/2}\right|\phi \right\rangle \neq 0 $.
		\end{itemize}
	\end{definition}

Hence, if $ \mathcal{W} $ is {\it resonant}, then $ - \Delta + \mathcal{W} $ has a {\it zero-energy resonance}. If in addition 
\begin{equation}\label{eq:nu}
  \nu(\sigma) = 1 +  \beta  \sigma + o(\sigma),
\end{equation}
for some $ \beta \in \mathbb{R} $, then \cite[Thm. 1.2.5]{albeverio2005sqm}
\begin{equation}
  \mathcal{K}_{\beta,\sigma} \xrightarrow[\sigma \to 0]{\left\| \: \cdot \: \right\| - \mathrm{ res}} \mathcal{H}_{\beta},
\end{equation}
where $ \mathcal{H}_{\beta} $ is defined in~\eqref{eq:eff ham} and $ \left\| \: \cdot \: \right\| - \mathrm{ res} $
is short for norm resolvent convergence.

Obviously, by standard results (see, \emph{e.g.}, \cite[Thm. 2.16]{kato1966gmw}), the one-parameter
unitary group $ e^{- i \mathcal{K}_{\beta,\sigma} t} $ generated by $ \mathcal{K}_{\beta,\sigma} $ strongly converges to $
e^{i \mathcal{H}_\beta t} $, when $ \sigma \to 0 $. Whether the same strong convergence holds true when $ \beta $ depends
on time, and therefore the propagator becomes a two-parameter unitary group, is not a consequence of some
general result of operator theory: there are indeed adaptations of the aforementioned results to
time-dependent operators~\cite{sloan1981,yoshida1974}, but they typically work only for bounded operators
or in presence of a common core independent of time (which is not the case for $ H_{\beta(t)} $). In the
Theorem below we fill this gap for time-dependent point interactions. We believe this result might be
of interest on its own.

\begin{thm}[Time-dependent point interaction dynamics]
  \label{pro:approx}
  \mbox{}	\\
  Let $ \beta \in C^1(\mathbb{R}) $ and $ H_{\beta(t)} $ and $ \mathcal{K}_{\beta(t),\sigma} $ be defined as in \eqref{eq:eff ham} and
  \eqref{eq:happ}, respectively, with $ \nu(\sigma) $ given by \eqref{eq:nu} and with $ \beta(t) $ in place of $ \beta $. Let also $
  \mathcal{U}_{\mathrm{eff}}(t,s) $ and $ \mathcal{U}_{\sigma}(t,s) $, $ t, s \in \mathbb{R} $, be the two-parameter
  unitary groups generated by $ H_{\beta(t)} $ and $ \mathcal{K}_{\beta(t),\sigma} $, then
  \begin{equation}
    \label{eq:approx}
    \mathcal{U}_{\sigma}(t,s) \xrightarrow[\sigma \to 0]{\mathrm{s}} \mathcal{U}_{\mathrm{eff}}(t,s).
  \end{equation}
\end{thm}

\begin{remark}[Many-center point interactions]
  \label{rem:many}
  \mbox{}	\\
  The result in \cref{pro:approx} is proven for a single time-dependent point interaction at
  the origin, but we expect that it is possible to extended it to point interactions with finitely many
  centers. The proof can indeed be adapted to take into account also the off-diagonal terms appearing in the
  time-evolution generated by a Schr\"{o}dinger operator with many point interactions.
\end{remark}

\bigskip

\noindent
{\footnotesize {\bf Acknowledgments.} R.C., M.F. and M.O. are thankful to INdAM group GNFM for the financial
support, through the grant Progetto Giovani GNFM 2017 ``Dinamica quasi-classica del polarone''. M.C. and M.O. are especially
grateful to the Institut Mittag-Leffler, where part of this work was completed. M.F.\ has been supported by the  Swiss National Science Foundation via the grant ``Mathematical Aspects of
  Many-Body Quantum Systems'', and by the European Research Council
    (ERC) under the European Union’s Horizon 2020 research and innovation
    programme (ERC CoG UniCoSM, grant agreement n.724939).}

\section{Approximation of the Effective Dynamics}
\label{sec:approx}

The main goal of this Section is the proof of \cref{pro:approx}. We thus focus on the limiting dynamics
described in \cref{sec:effective}: as anticipated, the existence of such a dynamics was originally
proven in \cite{sayyaf1984}, although later other alternative approaches have been developed to deal with
the same problem (see, \emph{e.g.}, \cite{posilicano2007}). The key idea is to show that the
ansatz~\eqref{eq:ansatz} preserves the domain of the quadratic form associated to $
\mathcal{H}_{\beta(t)} $, which is independent of $ t $. Next, one has to prove that the evolved state still
satisfies the boundary condition in \eqref{eq:eff ham} and therefore belongs to $
\mathscr{D}(\mathcal{H}_{\beta(t)}) $. This second step is easy to obtain if $ \beta(t) $ is regular enough: by
direct inspection of the Volterra integral equation \eqref{eq:charge eq}, one can indeed prove that, if $
\beta \in C^2_{\mathrm{loc}}(\mathbb{R}) $, then $ q \in C^2_{\mathrm{loc}}(\mathbb{R}) $ as well \cite[Thm.\ 1]{sayyaf1984}, and
it is also easy to verify that the boundary condition is satisfied. In the general case, the argument is
more involved but still the existence of the dynamics, as a two-parameter unitary group, can be proven for any
$ \beta \in L^{\infty}_{\mathrm{loc}}(\mathbb{R}) $ \cite[Thm.\ 2]{sayyaf1984}.

An important intermediate step towards \cref{pro:approx} is the approximation of the time-dependent
dynamics generated by $ \mathcal{H}_{\beta(t)} $ in terms of a product of time-independent ones, in the spirit
of \cite[Thm.\ 1, p.\ 432]{yoshida1974}. The idea is to split the interval $ [s,t] $ into smaller intervals and replace in each of those the propagator of $ \mathcal{H}_{\beta(t)} $ with the one associated to $ \mathcal{H}_{\beta_*(t)} $, $  \beta_*(t) $ being the step function approximation of $ \beta(t) $. The proof is then divided into three parts:
\begin{itemize}
	\item first, we show that the dynamics generated by $ \mathcal{H}_{\beta_*(t)} $ strongly converges to $ \mathcal{U}_{\mathrm{eff}} $ (\cref{lem: approx 1});
	\item next, we discuss the approximation of $ \mathcal{H}_{\beta_*(t)} $ and its unitary group in terms of the Schr\"{o}dinger operators \eqref{eq:happ} with rescaled potentials $ \mathcal{W}_{{\beta_*(t)},\sigma} $ (\cref{lem: approx 3}); 
	\item finally, we show that we can undo the step approximation of $ \beta(t) $ at the level of the dynamics generated by the approximant Schr\"{o}dinger operators $ \mathcal{K}_{\beta(t), \sigma} $ (\cref{lem: approx 2}).
\end{itemize}
All the results are then collected at the end of the Section to complete the
proof of \cref{pro:approx}. We remark that, as it is customary in the
  investigations of such questions, all the steps above are verified on a dense subset of the Hilbert space, which is given in our case by functions belonging to 
\beq
	\label{eq: domainD}
	\mathcal{D} : = \lf\{ \psi \in C^{\infty}_0(\R^3) \: \big| \: \psi(\xv) = \OO(|\xv|^{a}), \mbox{ as } |\xv| \to 0 \ri\},	\qquad \mbox{for some } a > \tx\frac{5}{2},
\eeq
which is obviously a dense set for both $  \mathcal{K}_{\beta, \sigma} $ and $ \mathcal{H}_{\beta} $, for any $ \beta \in \R $ and $ \sigma > 0 $.
Each result is then extended to the whole Hilbert space by density and unitarity of the propagators. In the following we do not track the dependence of the constants appearing in the statements on, {\it e.g.}, $ \beta $ or time $ t $, whenever the constant is finite for parameters varying on bounded intervals, as it is always the case in our setting.

Let us then fix the {initial time $ s \in \R $. Then, for any $ t \in \R $, with $ t > s $, we} divide {the interval $ [s,t] $} into $ n \in \mathbb{N} $ smaller intervals
\begin{equation}
  \label{eq:tj}
  I_j : = \left[ t_j, t_{j+1} \right), \qquad t_j : = s + \frac{j (t-s)}{n},
\end{equation}
with $ j = 0, \ldots, n-1 $. Next, we define $ \beta_{*} \in L^{\infty}(\mathbb{R}) $ as the step function
\begin{equation}
  \label{eq:betan}
  \beta_{*}(t) : = \beta(t_j),	\qquad \mbox{for } t \in I_j.
\end{equation}
Thanks to \cite[Thm.\ 2]{sayyaf1984} already mentioned above, the dynamics $
\mathcal{V}_n(t,s) $ generated by $ \mathcal{H}_{\beta_{*}(t)} $ exists as a two-parameter unitary group and, if
the initial state $ \psi_s $ at time $ s $ belongs to $ H^2(\mathbb{R}^3 \setminus \{0 \}) $, it can be represented as
in~\eqref{eq:ansatz}, \emph{i.e.},
\begin{equation}
  \label{eq:ansatz approx}
  \left( \mathcal{V}_{n}(t,s) \psi_s \right)(\mathbf{x}) = \left( U_0(t-s) \psi_s \right)(\mathbf{x}) + i \int_s^t \mathrm{d} \tau \: U_0(t - \tau; \mathbf{x}) q_n(\tau),
\end{equation}
where the charge $ q_n $ now solves the Volterra equation
\begin{equation}
  \label{eq:qapprox eq}
  q_n(t) + 4 \sqrt{\pi i} \int_s^t \mathrm{d} \tau \: \frac{\beta_{*}(\tau) q_n(\tau)}{\sqrt{t - \tau}} = 4 \sqrt{\pi i} \int_s^t \mathrm{d} \tau \: \frac{1}{\sqrt{t - \tau}} \left(U_0(\tau - s) \psi_s \right)(0).
\end{equation}
Note that $ \mathcal{V}_{n}(t,s) $ can be equivalently represented as
\begin{equation}
  \label{eq:uapprox equiv}
  \mathcal{V}_{n}(t,s) = e^{- i\mathcal{H}_{\beta_n}(t - t_{n-1})} e^{- i\mathcal{H}_{\beta_{n-1}}(t_{n-1} - t_{n-2})} \cdots e^{- i\mathcal{H}_{\beta_{0}}(t_1 - s)}.
\end{equation}

\begin{proposition}[Step function approximation of $ \mathcal{U}_{\mathrm{eff}}(t,s) $]
  \label{lem: approx 1}
  \mbox{}	\\
  Let $ \beta(t) \in C^1(\mathbb{R}) $. Then, for any $ s, t \in \mathbb{R} $,
  \begin{equation}
    \label{eq: approx}
    \mathcal{V}_{n}(t,s) \xrightarrow[n \to + \infty]{\mathrm{s}} \mathcal{U}_{\mathrm{eff}}(t,s).
  \end{equation}
\end{proposition}
	
\begin{proof}
  	The idea is to prove the result for a dense subset of initial states $ \psi_s $ and then extend it to the rest of the Hilbert space by density. Let us thus assume that $ \psi_s \in C^{\infty}_0(\mathbb{R}^3\setminus\{0\}) $, in which case the representation \eqref{eq:ansatz approx} holds true with $ q_n \in L^{\infty}_{\mathrm{loc}}(\mathbb{R}) $ solving \eqref{eq:qapprox eq}, as proven in \cite{sayyaf1984}. {Furthermore, the independence on $ n $ of the forcing term (term on the r.h.s. of \eqref{eq:qapprox eq}) and the uniform boundedness of $ \beta_{*} $ also imply that the $ L^{\infty} $ norm of $ q_n(t) $ on any finite interval is bounded uniformly in $ n \in \N $ (see also the general theory of Volterra integral equations in \cite{miller1971}).}
		
  By bootstrap, it is easy to see that, if $ q_n \in L^{\infty}_{\mathrm{loc}}(\mathbb{R}) $, then also $ q_n \in W^{1,1}_{\mathrm{loc}}(\mathbb{R}) $. Indeed, the forcing term can be easily seen to be differentiable, under the assumptions made: thanks to the properties of the Abel $1/2$-operator {\cite{GV}}, one has
  \begin{displaymath}
    \frac{\mathrm{d}}{\mathrm{d} t} \int_s^t \mathrm{d} \tau \: \frac{1}{\sqrt{t - \tau}} \left(U_0(\tau - s) \psi_s \right)(0) =  \int_s^t \mathrm{d} \tau \: \frac{1}{\sqrt{t - \tau}} \frac{\mathrm{d}}{\mathrm{d} \tau} \left(U_0(\tau - s) \psi_s \right)(0),
  \end{displaymath}
  and
  \begin{displaymath}
    \bigg| \frac{\mathrm{d}}{\mathrm{d} \tau} \left(U_0(\tau - s) \psi_s \right)(0) \bigg| = \frac{1}{(2\pi)^{3/2}} \bigg| \int_{\mathbb{R}^3} \mathrm{d} \mathbf{k} \: k^2 e^{-i k^2(\tau - s)} \widehat{\psi_s}(\mathbf{k}) \bigg|
    \leqslant C \int_{\mathbb{R}^3} \mathrm{d} \mathbf{k} \: k^2 \left| \widehat{\psi_s}(\mathbf{k}) \right| \leqslant C,
  \end{displaymath}
  so that {the derivative of the forcing term is bounded on finite intervals. Note that, as a direct consequence, the $ W^{1,1} $ norm of $ q_n $ on finite intervals is bounded uniformly in $ n \in \N $.}
  
  Next, we prove that the sequence $ q_n $ pointwise converges to $ q $, as $ n \to + \infty $: taking the difference between the equations \eqref{eq:charge eq} and \eqref{eq:qapprox eq} and
  setting $ \chi_n : = q - q_n $ for short, we obtain
  	\begin{multline*}
    		\chi_n(t) = 4 \sqrt{\pi i} \int_s^t \mathrm{d} \tau \: \frac{\beta_{*}(\tau) q_n(\tau) - \beta(\tau) q(\tau)}{\sqrt{t - \tau}} \\
    		= 4 \sqrt{\pi i} \int_s^t \mathrm{d} \tau \: \frac{\left(\beta_{*}(\tau) - \beta(\tau) \right) q_n(\tau)}{\sqrt{t - \tau}} - 4 \sqrt{\pi i} \int_s^t \mathrm{d} \tau \: \frac{\beta(\tau) \chi_n(\tau)}{\sqrt{t - \tau}} = : F_n(t) - 4 \sqrt{\pi i} \int_s^t \mathrm{d} \tau \: \frac{\beta(\tau) \chi_n(\tau)}{\sqrt{t - \tau}}\; .
  	\end{multline*}
  	It follows that $ \chi_n \in C(\mathbb{R}) $ (recall that both $ q $ and $ q_n $ belong to $   W^{1,1}_{\mathrm{loc}}(\mathbb{R}) $) is a solution of the Volterra equation
  	\begin{equation}
   		\label{eq:charge eq diff}
    		\chi_n(t) + 4 \sqrt{\pi i} \int_s^t \mathrm{d} \tau \: \frac{\beta(\tau)  \chi_n(\tau)}{\sqrt{t - \tau}} = F_n(t).
  \end{equation}	
  	Furthermore,
  	\begin{equation}
    		F_n(t) \xrightarrow[n \to + \infty]{} 0,
  	\end{equation}
  	pointwise, since
  	\begin{displaymath}
    		\left| F_n(t) \right| \leqslant C \sup_{\tau \in [s,t]} \left| q_n(\tau) \right| \: \int_s^t \mathrm{d} \tau \: \frac{\left| \beta_{*}(\tau) - \beta(\tau) \right|}{\sqrt{t - \tau}}  \xrightarrow[n \to + \infty]{} 0,
	\end{displaymath}
  	because
  	\begin{equation}
   		\beta_{*}  \xrightarrow[n \to + \infty]{L^p}  \beta,	\qquad		\mbox{for any } 1 \leqslant p < + \infty,
  	\end{equation}
  	and $ q_n $ is uniformly bounded. Hence, it just remains to observe that the l.h.s.\ of \eqref{eq:charge eq diff} is invertible for $ t $ small enough, since
  	\begin{equation}
    		\label{eq:est 1}
    		\left| \int_s^t \mathrm{d} \tau \: \frac{\beta(\tau)  \chi_n(\tau)}{\sqrt{t - \tau}} \right| \leqslant C \sup_{\tau \in [s,t]} \left| \chi_n(\tau) \right| \sqrt{t - s} \xrightarrow[s \to t]{} 0,
  	\end{equation}
  	which implies that the l.h.s.\ of \eqref{eq:charge eq diff} can be rewritten as $ \left( 1 + J \right) {\chi_n} $, with $ J: C_{\mathrm{loc}}(\mathbb{R}) \to C_{\mathrm{loc}}(\mathbb{R}) $ and such that $ \left\| J \right\|_{C \to C} < 1  $, if $ t -s $ is small enough. Hence, we get $ \chi_n = ( 1 + J)^{-1} F_n $, with $ F_n \to 0 $ and $  ( 1 +J)^{-1} $ bounded from $ C_{\mathrm{loc}} (\mathbb{R}) $ to $ C_{\mathrm{loc}}(\mathbb{R}) $. We conclude that $ \chi_n \to  0 $ in $ C_{\mathrm{loc}}(\mathbb{R}) $, if $ t - s $ is small enough.		
  
  The argument above guarantees that $ \chi_n \to 0 $ in an interval $ [s,t_1] $, where $ t_1 $ is independent of $ n $ and determined only by the condition that the r.h.s.\ of \eqref{eq:est 1} is strictly smaller  than $ \left\| \chi_n \right\|_{\infty} $. As such, $ t_1 $ depends only on $ s $ and $ \left\| \beta \right\|_{\infty} $. Hence it is not difficult to see that, picking $ n $ large enough and exploiting the convergence to $ 0 $ of $ \chi_n $ in $ [s,t_1] $, it is possible to reproduce the argument in the interval $ [s, 2t_1] $. A  bootstrap then gives pointwise convergence to $ 0 $ of $ \chi_n $ in any finite interval $ [s,t] $. In the last bootstrap, it is crucial that the considered interval is relatively compact.
  
  {Next, exploiting that the $ W^{1,1} $ and $ L^{\infty} $ norms of $ q_n $ are bounded uniformly in $ n $, we use the charge equation \eqref{eq:qapprox eq} once more, to show that the convergence of $ q_n $ to $ q $ is actually uniform on bounded intervals. Indeed, we claim that the sequence of functions $ \lf\{ q_n \ri\}_{n \in \N} $ is uniformly equicontinuous: for any $ 0 < t' < t $, we get}
 	\beq
 		\label{eq: approx 1 proof 1}
 		{q_n(t) - q_n(t') = - 4 \sqrt{\pi i} \int_{t'}^t \mathrm{d} \tau \: \frac{\beta_{*}(\tau) q_n(\tau)}{\sqrt{t - \tau}} + 4 \sqrt{\pi i} \int_{t'}^t \mathrm{d} \tau \: \frac{1}{\sqrt{t - \tau}} \left(U_0(\tau - s) \psi_s \right)(0).}
	\eeq
	{Now, the second term on the r.h.s. is obviously uniformly continuous, since we have shown that the function is actually $ C^1 $. Hence, we have only to consider the first term on the r.h.s. of \eqref{eq: approx 1 proof 1}. However, by the uniform boundedness of both $ \beta_{*} $ and $ q_n $, we can bound}
	\bdm
		{\lf| \int_{t'}^t \mathrm{d} \tau \: \frac{\beta_{*}(\tau) q_n(\tau)}{\sqrt{t - \tau}} \ri| \leq C  \sqrt{t - t'},}
	\edm
	{which is independent of $ n \in \N $, so implying uniform equicontinuity of the sequence. A direct application of the Ascoli-Arzel\'{a} theorem then yields uniform convergence of $ q_n $ to  $ q $ on finite intervals.}
		
  Moreover, using \eqref{eq:ansatz} and \eqref{eq:ansatz approx}, we get
  \begin{equation}
    \label{eq:unit diff}
    \left\| \left( \mathcal{U}_{\mathrm{eff}}(t,s) - \mathcal{V}_{n}(t,s) \right) \psi_s \right\|_2 \leqslant C \sup_{\tau \in [s,t]} \left| q(\tau) - q_n(\tau) \right|,
  \end{equation}
  and therefore the strong convergence of the propagators follows from the uniform convergence of the
  charges, for $ \psi_s \in C_0^{\infty}(\mathbb{R}^3\setminus\{0\}) $. 
\end{proof}

We recall that for any time-independent $ \beta \in \mathbb{R} $, the family of operators $ \mathcal{K}_{\beta,\sigma} $ defined in \eqref{eq:happ} and given by $\mathcal{K}_{\beta,\sigma} = - \Delta + \mathcal{W}_{{\beta},\sigma}(\mathbf{x}) $, where $ - \Delta + \mathcal{W} $ is assumed to have a zero-energy resonance, $ \mathcal{W}_{{\beta},\sigma}(\mathbf{x}) = \nu(\sigma) \sigma^{-2} \mathcal{W}(\mathbf{x}/\sigma) $ and $ \nu(\sigma) = 1 +
\beta \sigma + o(\sigma) $, converges in norm resolvent sense to $ \mathcal{H}_{\beta} $, as $ \sigma \to 0 $ \cite[Thm. 1.2.5]{albeverio2005sqm}.  A by-product of this result is the strong convergence of the
corresponding unitary operators obtained by the piecewise approximations of $ \beta(t) $: in the following we denote by $ \mathcal{V}_{n,\sigma}(t,s) $ the two-parameter unitary group associated with the time evolution generated by $ \mathcal{K}_{\beta_*(t), \sigma} $, {\it i.e.}, such that
\beq
	i \partial_t \mathcal{V}_{n,\sigma}(t,s) = \mathcal{K}_{\beta_*(t), \sigma} \mathcal{V}_{n,\sigma}(t,s),
\eeq
and study the strong limit $ \sigma \to 0 $ of such operators. Note that, by \eqref{eq:betan}, one obviously gets that
\begin{equation}
  	\label{eq: uapproxs}
  	\mathcal{V}_{n,\sigma}(t,s) : = e^{- i\mathcal{K}_{\beta_{n-1},\sigma}(t - t_{n-1})} e^{- i\mathcal{K}_{\beta_{n-2},\sigma}(t_{n-1} - t_{n-2})} \cdots e^{- i\mathcal{K}_{\beta_{0},\sigma}(t_1 - s)}.
\end{equation}
 
Before doing that, we have however to formulate a technical lemma concerning the approximation
of functions evolved with $ e^{-i \mathcal{H}_{\beta} (t-s)} $ or, more in general, with $ \mathcal{V}_n(t,s) $, which is simply a consequence of the density of $ \mathcal{D} $ in $ L^2(\R^3) $. We state a quantitative bound however for concreteness.
 
\begin{lemma}
  \label{lem: approx smooth}
  \mbox{}	\\
  Let $ \chi \in L^2(\R^3) $, such that $ \left\| \chi \right\|_2 = 1 $. Let also $ s < t \in \R $ varying in a compact set. Then, there exists a sequence of functions $ \lf\{ \chi_m \ri\}_{m \in \N} \subset \mathcal{D} $,  and a finite constant $ C > 0 $ independent of $ n $, such that $ \left\| \chi_m \right\|_2 \leqslant \left\| \chi \right\|_2 $ and
  \begin{equation}
    \label{eq: approx smooth}
    \left\| \mathcal{V}_n(t,s)  \chi - \chi_m \right\|_2 \leqslant C m^{-2/3}.
  \end{equation}
\end{lemma}
	
\begin{proof}
  	The first simple observation is that, by direct inspection of the charge equation \eqref{eq:ansatz approx}, which reduces to \eqref{eq:charge eq} in each interval $ [t_j,t_{j+1}] $ with $ \beta = \beta(t_j) $ and initial datum $ \psi(t_j) $, one gets that the charge $ q(t) $ is uniformly bounded on compact sets irrespective of $ n \in \N $.
		
  	Then, it is sufficient to set for $ m \in \mathbb{N}_0 $
  	\begin{equation}
    		\label{eqp: chim}
    		\chi_m(\mathbf{x}) : = e^{-\frac{1}{m x}} \left( \mathcal{V}_n(t,s) \chi \right) (\mathbf{x}).
  \end{equation}
  	We claim that such a sequence satisfies \eqref{eq: approx smooth}. By the characterization of the
  operator domain, we know that $ \mathcal{V}_{n}(\tau,s) \chi \in \mathscr{D}(\mathcal{H}_{\beta_*(\tau)}) $, {\it i.e.}, we can decompose for any $ \tau \in [s,t] $
  	\begin{equation}
    		\left( \mathcal{V}_n(\tau,s) \chi \right)(\mathbf{x}) = \phi_{{\tau}}(\mathbf{x}) + \frac{q(\tau) e^{-x}}{4\pi x},
  	\end{equation}
  	where $ q(\tau) $ is a solution of \eqref{eq:ansatz approx} and $ \phi_{\tau} \in H^2(\mathbb{R}^3) $. Note that we have used a slightly different domain decomposition than the one given in \eqref{eq:eff ham}: starting from the domain given there, one simply recovers the expression above adding and subtracting $ \frac{q(\tau) e^{-x}}{4\pi x} $ and observing that $ \frac{e^{-x}}{x} \in L^2(\mathbb{R}^3) $. Hence,
  	\begin{multline}
    		\left\| \mathcal{V}_n(t,s) \chi - \chi_m \right\|_2^2 = \int_{\mathcal{B}_R} \mathrm{d} \mathbf{x} \: \left(1 - e^{-\frac{1}{m x}} \right)^2 \left| \phi_{t}(\mathbf{x}) + \frac{q(t) e^{-x}}{4\pi x} \right|^2 	\\
    + \int_{\mathcal{B}_R^c} \mathrm{d} \mathbf{x} \: \left(1 - e^{-\frac{1}{m x}} \right)^2 \left|
      \phi_{t}(\mathbf{x}) + \frac{q(t) e^{-x}}{4\pi x} \right|^2
  	\end{multline}
  	where $ R > 0 $ is a parameter to be chosen later.		
  	Next, we estimate
 	 \begin{multline}
    		\int_{\mathcal{B}_R} \mathrm{d} \mathbf{x} \: \left(1 - e^{-\frac{1}{m x}} \right)^2 \left| \phi_{t}(\mathbf{x}) + \frac{q(t) e^{-x}}{4\pi x} \right|^2 \leqslant C m^{-3} \int_{\mathcal{B}_{m R}} \mathrm{d} \mathbf{x} \left[ \left| \phi_{t}(\mathbf{x}) \right|^2 + \frac{ m^2|q(t)|^2}{x^2} \right] \\
    		\leqslant C \left( R^3 + R \right),
  	\end{multline}
  	by \eqref{eqp: charge est 2} and the fact that $ \phi_{t} \in H^2(\mathbb{R}^3) $. On the
  	other hand,
  	\begin{multline}
    		\int_{\mathcal{B}_R^c} \mathrm{d} \mathbf{x} \: \left(1 - e^{-\frac{1}{m x}} \right)^2 \left| \phi_{t}(\mathbf{x}) + \frac{q(t) e^{-x}}{4\pi x} \right|^2 \leqslant \frac{C}{m^{2}} \int_{\mathcal{B}_R^c} \mathrm{d} \mathbf{x} \: \frac{1}{x^2} \left[ \left| \phi_{t}(\mathbf{x}) \right|^2 + \frac{|q(t)|^2}{x^2} \right] \\
    \leqslant C m^{-2} \left( R^{-2} + R^{-1} \right).
  	\end{multline}
		
  	Altogether the above estimates yield
  	\beq
    		\left\| \mathcal{V}_n(\tau,s)  \chi - \chi_m \right\|_2^2 \leqslant C \left[ R^3 + m^{-2} R^{-2} + ( R + m^{-2} R^{-1} \right] 
    		\leqslant C m^{-2/3}
 	\eeq
  	after an optimization over $ R $, {\it i.e.}, taking $ R = m^{-2/3} $.
\end{proof}
	
Before discussing the convergence of $ \mathcal{V}_{n,\sigma} $ to $ \mathcal{V}_n $, we need one more technical
result: we are going to show that a strong estimate over a dense set of the difference between the unitary
evolutions generated by $ \mathcal{H}_{\beta} $ and $ \mathcal{K}_{\beta,\sigma} $ is sufficient to control the
difference between the unitaries $ \mathcal{V}_{n}(t,s) $ and $ \mathcal{V}_{n,\sigma}(t,s) $ in strong sense. The result is going to be used in next \cref{lem: approx 3}, where we are going to show that the bound \eqref{eq: hp lemma} holds true on a dense set with some explicit error $ \delta $. The result below then allows to translate the bound \eqref{eq: hp lemma} into an estimate of the difference between the unitary groups on the same dense set.
	
\begin{lemma}
  \label{lem: strong weak}
  \mbox{}	\\
  Let $ t_0 = s  <  t_1  < \cdots  <  t_n = t $ be a partition of the bounded interval $ [s,t] $ as in
  \eqref{eq:tj}. Let $ \psi \in \mathcal{D} $, with
  $ \left\| \psi \right\|_2 \leqslant 1 $. Assume
  that there exists $ \delta>0 $ such that, for all $\phi\in \mathcal{D}$ and for all finite $ \tau $,
  \begin{equation}
    \label{eq: hp lemma}
    \left\| \left( e^{- i \mathcal{H}_{\beta} \tau} - e^{- i \mathcal{K}_{\beta,\sigma} \tau} \right) \phi \right\|_2 \leqslant \delta\lVert \phi  \rVert_2^{}\; .
  \end{equation}
  Then,
  \begin{equation}
    \label{eq: strong weak}
    \lf\| \left( \mathcal{V}_{n}(t,s) - \mathcal{V}_{n,\sigma}(t,s) \right) \psi \ri\|_2 \leqslant  n \delta + \mathcal{O}(n^{-1}).
  \end{equation}
\end{lemma}
\begin{proof}
  The result is proved iteratively on the quantity
  \begin{displaymath}
    \left\lVert \left(  \mathcal{V}_{n}(t,s) - \mathcal{V}_{n,\sigma}(t,s) \right) \psi  \right\rVert_{2}^{}\;.
  \end{displaymath}
  We have that
  \begin{multline*}
    \left\lVert \left(  \mathcal{V}_{n}(t,s) - \mathcal{V}_{n,\sigma}(t,s) \right) \psi  \right\rVert_{2}^{}\leq \left\lVert \bigl(e^{i\mathcal{H}_{\beta_{n-1}}(t-t_{n-1})}-e^{i\mathcal{K}_{\beta_{n-1},\sigma}(t-t_{n-1})}\bigr)\mathcal{V}_{n}(t_{n-1},s)\psi  \right\rVert_{2}^{}\\ + \left\lVert e^{i\mathcal{K}_{\beta_{n-1},\sigma}(t-t_{n-1})}\bigl(\mathcal{V}_{n}(t_{n-1},s)-V_{n,\sigma}(t_{n-1},s)\bigr)\psi  \right\rVert_{2}^{}\; .
  \end{multline*}  
  Now, we can apply \cref{lem: approx smooth} {with} $ \chi =
  \mathcal{V}_{n}(t_{n-1},s)\psi$, and we get a sequence of functions $ \chi_m \in
  \mathcal{D} $, such that $ \left\| \chi_m \right\|_2 \leqslant \left\|
    \mathcal{V}_{n}(t_{n-1},s)\psi \right\|_2 \leqslant 1 $ and
  \begin{displaymath}
    \mathcal{V}_{n}(t_{n-1},s)\psi  - \chi_m = \mathcal{O}(m^{-2/3})\;,
  \end{displaymath}
  for any $ m \in \mathbb{N}_0 $. In particular, if we take $ m = n^3 $, we obtain
  \begin{equation}
    \mathcal{V}_{n}(t_{n-1},s)\psi  - \chi_m = \mathcal{O}(n^{-2})\;.
  \end{equation}
  Therefore,
  \begin{multline*}
    \left\lVert \left(  \mathcal{V}_{n}(t,s) - \mathcal{V}_{n,\sigma}(t,s) \right) \psi  \right\rVert_{2}^{}\leq \left\lVert \bigl(e^{i\mathcal{H}_{\beta_{n-1}}(t-t_{n-1})}-e^{i\mathcal{K}_{\beta_{n-1},\sigma}(t-t_{n-1})}\bigr)\chi_m  \right\rVert_{2}^{}\\
    + \left\lVert \bigl(\mathcal{V}_{n}(t_{n-1},s)-V_{n,\sigma}(t_{n-1},s)\bigr)\psi  \right\rVert_{2}^{}+ \mathcal{O}(n^{-2})\; .
  \end{multline*}
  Now, using the assumption \eqref{eq: hp lemma}, we get
  \begin{equation*}
    \left\lVert \left(  \mathcal{V}_{n}(t,s) - \mathcal{V}_{n,\sigma}(t,s) \right) \psi  \right\rVert_{2}^{}\leq \left\lVert \bigl(\mathcal{V}_{n}(t_{n-1},s)-V_{n,\sigma}(t_{n-1},s)\bigr)\psi  \right\rVert_{2}^{}+ \delta + \mathcal{O}(n^{-2})\; .
  \end{equation*}
  Iterating such reasoning for the remaining $n-1$ intervals yields
  \begin{equation*}
    \left\lVert \left(  \mathcal{V}_{n}(t,s) - \mathcal{V}_{n,\sigma}(t,s) \right) \psi  \right\rVert_{2}^{}\leq n\delta + \mathcal{O}(n^{-1})\; ,
  \end{equation*}
  thus concluding the proof.
\end{proof}

The technical results proven in \cref{lem: approx smooth,lem: strong weak} are used to
prove the convergence of $ \mathcal{V}_{n,\sigma}(t,s) $ to $ \mathcal{V}_{n}(t,s) $, which in turn is going to
play a key role in the proof of \cref{pro:approx}. We are going to use a known representation formula for the propagator, provided in next Lemma.

\begin{lemma}
	\label{lem: pazy}
	\mbox{}	\\
	For any $ t \in \R $ finite and for any $ \psi \in \mathscr{D}\left(\mathcal{H}_{\beta}\right) \cap \mathscr{D}\left(\mathcal{K}_{\sigma, \beta} \right) $, with $ \lf\| \mathcal{K}_{\sigma, \beta} \psi \ri\| \leq C $,
	\begin{multline}
    		\label{eqp: approx est 2}
    		{\left\| \left[ e^{- i\mathcal{H}_{\beta} t} - e^{-i \mathcal{K}_{\sigma, \beta} t} \right] \psi \right\|_2 =  \left(\frac{k}{\lvert t  \rvert_{}^{}}\right)^k \sup_{\chi \in L^2(\mathbb{R}^3), \left\| \chi \right\|_2 \leqslant 1} \left| \left\langle \chi\left|\left[ \left( \mathcal{H}_{\beta} - \tx\frac{i k}{t} \right)^{-k} -\left(\mathcal{K}_{\sigma, \beta} - \tx\frac{i k}{t} \right)^{-k} \right]\psi\right. \right\rangle_2\right|} \\
    		{+ \lvert t  \rvert_{}^{} o_k(1),}
  	\end{multline}
  	{unformly in $ \sigma \in [0,1] $.}
\end{lemma}

\begin{proof}
	The starting point is the representation formula \cite[Thm.\ 8.3, Ch.\ 1]{pazy1983ams} for a unitary group $ e^{-it A} $ generated by a self-adjoint operator $ A $: for any $ \psi \in L^2(\mathbb{R}^3) $, one has the following convergence in norm, uniformly w.r.t. $t$ on bounded intervals:
  	\begin{equation}
    		e^{-it A} \psi = \lim_{k \to + \infty} \left( 1 + \frac{i t A}{k} \right)^{-k} \psi = \lim_{k \to + \infty} \left( \frac{k}{it} \right)^k \left( A - \frac{i k}{t} \right)^{-k} \psi\; .
 	\end{equation}
  	Since the vector $\psi$ belongs to the intersection of the domains of the two operators, it is possible to write a more explicit bound on the error:
  	\begin{equation}
		\label{eq: pazy proof 1}
    		\begin{split}
      	\lf(e^{- it\mathcal{H}_{\beta} }  - \left( \tx\frac{k}{i t} \right)^k \left( \mathcal{H}_{\beta} - \tx\frac{i k}{t} \right)^{-k} \ri) \psi =-i\int_0^t \lf(e^{- i\tau\mathcal{H}_{\beta}}  - \left( 1 + \tx\frac{i \tau }{k}\mathcal{H}_{\beta} \right)^{-k-1} \ri) \mathcal{H}_{\beta}\psi \, \mathrm{d}\tau\\= |t| o_k(1),
    	\end{split}
  	\end{equation}
  	where $o_k(1)\in L^2(\mathbb{R}^3)$ converges to zero as $k\to \infty$ uniformly with respect to $t$ and $\beta$ on bounded  intervals. 
  
  	{Let us prove that a similar estimates holds true for $ \mathcal{K}_{\beta,\sigma} $: following \cite{pazy1983ams}, we write}
  	\bmln{
  		{\lf\| \lf(e^{- it \mathcal{K}_{\sigma, \beta} }  - \left( \tx\frac{k}{i t} \right)^k \left( \mathcal{K}_{\sigma, \beta} - \tx\frac{i k}{t} \right)^{-k} \ri) \psi \ri\|_2 = \frac{k^{k+1}}{k!} \lf\| \int_0^{+\infty} \diff \xi \: \xi^k e^{- k \xi} \lf( e^{- it \mathcal{K}_{\sigma, \beta} }  -e^{-i \xi t \mathcal{K}_{\sigma, \beta}}  \ri) \psi \ri\|_2}	\\
  		{\leq C \frac{k^{k+1}}{k!} \int_0^{+\infty} \diff \xi \: \xi^k e^{- k \xi} = C \frac{1}{k!} \int_0^{+\infty} \diff \xi \: \xi^k e^{- \xi} = C},
	}
	{so that}
	\beq
		{\lf\| \lf(e^{- it \mathcal{K}_{\sigma, \beta} }  - \left( \tx\frac{k}{i t} \right)^k \left( \mathcal{K}_{\sigma, \beta} - \tx\frac{i k}{t} \right)^{-k} \ri) \psi \ri\|_2 \xrightarrow[k \to +\infty]{} 0,}
	\eeq
	{by dominated convergence and pointwise convergence to $ 0 $ of the integrand:}
	\bdm
		{\frac{\xi^k e^{-\xi}}{k!} \leq \frac{k^k e^{-k}}{k!} \leq \frac{C}{\sqrt{k}}  \xrightarrow[k \to + \infty]{} 0.}
	\edm 
	{Note that the convergence is uniform in $ \sigma $, since the estimates above are independent of $ \sigma $.}
	
	{Furthermore, for any $ \psi \in \dom\lf( \mathcal{K}_{\sigma, \beta}  \right) $,}
	\beq	
		{\lf\| \lf( \tx\frac{k}{i t}  \left( \mathcal{K}_{\sigma, \beta} - \tx\frac{i k}{t} \right)^{-1} - 1 \ri) \psi \ri\|_2 = \lf\|  \left( \mathcal{K}_{\sigma, \beta} - \tx\frac{i k}{t} \right)^{-1} \mathcal{K}_{\sigma, \beta}  \psi \ri\|_2 \leq C t/k \xrightarrow[k \to + \infty]{} 0,}
	\eeq
	{uniformly in $ \sigma $. This implies the analogue of \eqref{eq: pazy proof 1} for $   \mathcal{K}_{\sigma, \beta} $, via triangular inequality.}
\end{proof}

	The core of the proof of \cref{pro:approx} is next \cref{lem: approx 3}, where we show that the norm resolvent convergence of $ \mathcal{K}_{\sigma,\beta} $ to $ \mathcal{H}_{\beta} $ is in fact sufficient, via \cref{lem: approx smooth,lem: strong weak,lem: pazy} to deduce a quantitive estimate of the convergence on a suitable dense subset of the unitary groups $ \mathcal{V}_{n,\sigma} $ and $ \mathcal{V}_{n} $ as $ \sigma \to 0 $. It is precisely at this point that we need to restrict the argument to the set $ \mathcal{D} $ introduced in \eqref{eq: domainD}.
	
\begin{proposition}[Rescaled potential approximation of $ \mathcal{V}_{n}(t,s) $]
  \label{lem: approx 3}
  \mbox{}	\\
  For any finite $ t,s \in \mathbb{R} $ and for any $ \psi, \phi \in \mathcal{D} $, with $ \left\| \phi \right\|_2 = \left\| \psi \right\|_2 = 1 $, as $ \sigma \to 0 $,
  with  $\sigma n \ll 1$,
  \begin{equation}
    \label{eq: approx 3}
    \left| \left\langle \phi\left| \left( \mathcal{V}_{n,\sigma}(t,s) - \mathcal{V}_{n}(t,s) \right) \psi\right. \right\rangle \right| = \mathcal{O}(n^5 \sigma^{2}) + n^{-1} o_{\sigma}(1) + o_n(1).
  \end{equation}
\end{proposition}

\begin{proof}
  	The idea is to use \cref{lem: strong weak}, and thus focus on estimating the quantity
  	\begin{equation}
    		\label{eqp: approx est 1}
    		\left\| \left[ e^{- i\mathcal{H}_{\beta} \tau} - e^{-i \mathcal{K}_{\sigma, \beta} \tau} \right] \psi  \right\|_2		
  	\end{equation}
	for any $ \psi \in H^2(\mathbb{R}^3) $ such that $ \psi(0) = 0 $, $ \beta \in \mathbb{R} $ and
  	\begin{equation}
   		\tau \sim \frac{1}{n}\; .
  	\end{equation}
  	By assumption $  \psi \in \mathscr{D}\left(\mathcal{H}_{\beta}\right) \cap \mathscr{D}\left(\mathcal{K}_{\sigma, \beta} \right) $ and we are thus in position to apply \cref{lem: pazy}, yielding
  	\begin{multline}
    		\label{eqp: approx est 2}
    		\left\| \left[ e^{- i\mathcal{H}_{\beta} \tau} - e^{-i \mathcal{K}_{\sigma, \beta} \tau} \right] \psi \right\|_2 \\
    		= \left(\frac{k}{\lvert \tau  \rvert_{}^{}}\right)^k \sup_{\chi \in L^2(\mathbb{R}^3), \left\| \chi \right\|_2 \leqslant 1} \left| \left\langle \chi\left|\left[ \left( \mathcal{H}_{\beta} - \tx\frac{i k}{\tau} \right)^{-k} -\left(\mathcal{K}_{\sigma, \beta} - \tx\frac{i k}{\tau} \right)^{-k} \right]\psi\right. \right\rangle_2\right| + \lvert \tau  \rvert_{}^{} o_k(1),
  	\end{multline}
  	{uniformly in $ \sigma $.}
		
  Now we claim that, given two self-adjoint operators $ A $ and $ B $, $ z \in \rho(A) \cap \rho(B) $ and $ k \in \mathbb{N} $,
  if there exists some $ \delta_z < + \infty $ such that, for any $ \phi, \psi \in L^2(\mathbb{R}^3) $, such that $ \left\| \psi
  \right\|_2 \leqslant 1 $, $ \left\| \phi \right\|_2 \leqslant 1 $,
  \begin{equation}
    \label{eqp: condition}
    \left| \left\langle \phi\left|\left[ (A - z)^{-1} - (B-z)^{-1} \right] \psi\right. \right\rangle_2\right| \leqslant  \delta_z \left\| \phi \right\|_2 \left\| \psi \right\|_2,
  \end{equation}
  then
  \begin{equation}
    \label{eqp: conclusion}
    \left| \left\langle \phi\left|\left[ (A - z)^{-k} - (B-z)^{-k} \right] \psi\right. \right\rangle_2\right| \leqslant \frac{k \delta_z}{\left|\Im(z)\right|^{k-1}}.
  \end{equation}
  The result can be proven by induction writing
  \begin{multline*}
    (A - z)^{-k} - (B-z)^{-k} = \left( (A - z)^{-1} - (B-z)^{-1} \right) (A - z)^{-k+1} \\
    + (B - z)^{-1} \left((A - z)^{-k+1} - (B-z)^{-k+1} \right),
  \end{multline*}
  and using the inequalities
  \begin{displaymath}
    \left\| (A-z)^{-1} \right\| \leqslant \frac{1}{|\Im(z)|},	\qquad		\left\| (B-z)^{-1} \right\| \leqslant \frac{1}{|\Im(z)|},
  \end{displaymath}
  in the consequent estimate
  \begin{multline*}
    \left| \left\langle \phi\left|\left[ (A - z)^{-k} - (B-z)^{-k} \right] \psi\right. \right\rangle_2\right| \leqslant \left| \left\langle \phi\left|\left[ (A - z)^{-1} - (B-z)^{-1} \right] (A - z)^{-k+1} \psi\right. \right\rangle_2 \right| \\
    + \left| \left\langle (B - z^*)^{-1}\phi\left|\left[ (A - z)^{-k+1} - (B-z)^{-k+1}
          \right]\psi\right. \right\rangle_2\right| \\ \leqslant \frac{\delta_z}{\left|\Im(z)\right|^{k-1}} + \left| \left\langle (B -
        z^*)^{-1}\phi\left|\left[ (A - z)^{-k+1} - (B-z)^{-k+1} \right] \psi\right. \right\rangle_2 \right|,
  \end{multline*}
  which leads to \eqref{eqp: conclusion} by the induction hypothesis.
		
  {Therefore, using \eqref{eqp: conclusion} in \eqref{eqp: approx est 2}, we get (keeping in mind that $\psi$
  has norm one)
  \begin{equation}
    \label{eqp: approx 3}
    \left\| \left[ e^{- i\mathcal{H}_{\beta} \tau} - e^{-i \mathcal{K}_{\sigma, \beta} \tau} \right] \psi \right\|_2 \leqslant \frac{k^2 \delta_{n,k, \sigma}}{\lvert \tau  \rvert_{}^{}}  + \lvert \tau \rvert_{}^{} o_k(1),
  \end{equation}
  where $ \delta_{n,k, \sigma} $ is such that, for any normalized $ \phi, \psi \in L^2(\mathbb{R}^3) $,
  \begin{equation}
    \left| \left\langle \phi\left|\left[ \left( \mathcal{H}_{\beta} - \tx\frac{i k}{\tau} \right)^{-1} - \left(\mathcal{K}_{\sigma, \beta} - \tx\frac{i k}{\tau} \right)^{-1} \right] \psi\right. \right\rangle_2 \right| \leqslant  \delta_{n,k, \sigma}.
  \end{equation}
  Notice that we already know that such a $ \delta_{n,k, \sigma} $ does exist thanks to the norm resolvent
  convergence of $ \mathcal{K}_{\sigma, \beta} $ to $ \mathcal{H}_{\beta} $ stated in
  \cite[Thm. 1.2.5]{albeverio2005sqm} and, more precisely, $ \delta_{n, k,\sigma} \to 0 $, as $ \sigma \to 0 $, for fixed $ n,k
  \in \mathbb{N} $. However, we are going to estimate the dependence of $ \delta_{n,k,\sigma} $ on the parameters $ n $, $k$, and $
  \sigma $, showing that
  \begin{equation}
    \label{eqp: delta nks}
    \delta_{n,k, \sigma} =  \mathcal{O}\left(\sigma^2 \sqrt{nk} \right) + n^{-5/2} k^{-5/2} o_{\sigma}(1).
  \end{equation}}
  In fact, the result is proven by simply tracking down in \cite[Proof of Thm. 1.2.5]{albeverio2005sqm}
  the dependence of the remainders on the spectral parameter. With the same notation as in
  \cite{albeverio2005sqm}, we have (recall the definition \eqref{eq:nu} of $ \nu(\sigma) $)
  \begin{eqnarray}
    \label{eqp: resolvent 1}
    \left( \mathcal{H}_{\beta}+\lambda_n \right)^{-1} - \left( - \Delta + \lambda_n \right)^{-1} & =: & - \frac{1}{\beta + \frac{\sqrt{\lambda_n}}{\pi}} \left|G_{\lambda_n} \right\rangle \left\langle G_{\lambda_n}\right|,	\\	
    \label{eqp: resolvent 2}
    \left( \mathcal{K}_{\beta,\sigma}+\lambda_n \right)^{-1} - \left( - \Delta+\lambda_n \right)^{-1} & =: & - \sigma \nu(\sigma) A_{\sigma, n} \left(1 + B_{\sigma,n} \right)^{-1} C_{\sigma,n},	
  \end{eqnarray}
  where the operators $ A_{\sigma,n} $, $ B_{\sigma,n} $ and $ C_{\sigma,n} $ are defined in \eqref{eqp: operators}
  below, we have set for short
 \begin{equation}
     {\lambda_n : = \frac{i k}{\tau},}
  \end{equation}
  $ G_{\lambda}(\mathbf{x}) : = (-\Delta + \lambda)^{-1}(\mathbf{x})  $ is the Green function of the Laplacian, $ v =
  \sqrt{|\mathcal{W}|} $, $ u = \mathrm{sgn}(\mathcal{W}) \sqrt{|\mathcal{W}|} $ and $ \phi_0 $ is the $ L^2$ solution of the zero-energy
  equation
  \begin{equation}
    u (-\Delta)^{-1} v \phi_0 = - \phi_0,
  \end{equation}
  which is known to exist and being non-trivial thanks to the resonance condition, which also ensures that
  $ \left\langle v|\phi_0 \right\rangle \neq 0 $. Note that, since $ \lambda_n $ is purely imaginary ad its modulus diverges as $
  n, k \to + \infty $, the Green function $ G_{\lambda_n} $ belongs to $ L^2(\mathbb{R}^3) $ uniformly in $ n $ and $ k $. The
  operators $ A_{\sigma,n}, B_{\sigma,n} $ and $ C_{\sigma,n} $ are integral operators whose kernels are given by (see
  also \cite[Defs. (1.2.12) -- (1.2.14)]{albeverio2005sqm})
  \begin{eqnarray}
    \label{eqp: operators}
    A_{\sigma,n}(\mathbf{x},\mathbf{x}^{\prime}) & : = & G_{\lambda_n}(\mathbf{x} - \sigma \mathbf{x}^{\prime}) v(\mathbf{x}^{\prime}),	\\
    B_{\sigma,n}(\mathbf{x},\mathbf{x}^{\prime}) & : = & \nu(\sigma) u(\mathbf{x}) G_{\sigma^2 \lambda_n}(\mathbf{x} - \mathbf{x}^{\prime}) v(\mathbf{x}^{\prime}),	\\
    C_{\sigma,n}(\mathbf{x},\mathbf{x}^{\prime}) & : = & u(\mathbf{x}) G_{\lambda_n}(\varepsilon \mathbf{x} - \mathbf{x}^{\prime}).
  \end{eqnarray}
  At fixed $ n $ and $ k $ it is not difficult to see \cite[Lemma 1.2.2]{albeverio2005sqm} that
  \begin{eqnarray}
    \label{eqp: operators limit}
    A_{\sigma,n} & \xrightarrow[\sigma \to 0]{} & A_n := \left|G_{\lambda_n} \right\rangle \left\langle v\right|,		\label{eqp: A}\\
    B_{\sigma,n} & \xrightarrow[\sigma \to 0]{} & B := u (-\Delta)^{-1} v,			\label{eqp: B}\\
    C_{\sigma,n} & \xrightarrow[\sigma \to 0]{} & C_n := \left|u \right\rangle \left\langle G_{\lambda_n}\right|,	\label{eqp: C}
  \end{eqnarray}
  where $ u $ and $ v $ are meant as the multiplication operators by $ u $ and $ v $, respectively. The
  convergences above can be proven in Hilbert-Schmidt norm. Note that the operator $ B$ is independent of
  $ n $ and $ k $.
		
  In order to estimate the difference between the resolvent we use \eqref{eqp: resolvent 1} and
  \eqref{eqp: resolvent 2} to write
  \begin{multline*}
    \left\langle \phi\left|\left[ \left( \mathcal{H}_{\beta}+\lambda_n \right)^{-1} - \left( \mathcal{K}_{\beta,\sigma}+\lambda_n \right)^{-1} \right] \psi\right. \right\rangle_2 = -  \underbrace{\sigma \nu(\sigma) \left\langle \phi\left|\left( A_{\sigma, n} - A_n \right) \left(1 + B_{\sigma,n} \right)^{-1} C_{\sigma,n} \psi\right. \right\rangle_2}_{(1)} \\
    - \underbrace{\sigma \nu(\sigma) \left\langle \psi\left|A_n \left(1 + B_{\sigma,n} \right)^{-1} \left(C_{\sigma,n} - C_n \right)
          \psi\right. \right\rangle_2}_{(2)} \\
          - \underbrace{\left\langle \phi\left| \left[\sigma \nu(\sigma) A_n\left(1 +
              B_{\sigma,n}\right)^{-1} C_n - \textstyle\frac{1}{\beta + \frac{\sqrt{\lambda_n}}{\pi}} \left|G_{\lambda_n} \right\rangle \left\langle G_{\lambda_n}\right| \right]
         \psi\right. \right\rangle_2}_{(3)}.
  \end{multline*}
  Note that the quantity $ \beta + \frac{\sqrt{\lambda_n}}{\pi} $ is invertible because
  \begin{displaymath}
    \Im \left( \beta + \frac{\sqrt{\lambda_n}}{\pi} \right) = \frac{1}{\pi} \sqrt{\frac{k}{\tau_n}} \Im \left(\sqrt{i} \right) \neq 0,
  \end{displaymath}
  for any $ k, n $.  Terms (1) and (2) above are the easiest to bound and can in fact bounded in the very
  same way: since however the bound requires an estimate of the norm of $ B_{\sigma,n} $, which is also
  involved in term (3), we start from this last one. By \eqref{eqp: A} and \eqref{eqp: C},
  \begin{multline}
    \left\langle \phi\left|\left[ \sigma \nu(\sigma) A_n \left(1 + B_{\sigma,n} \right)^{-1}C_n - \textstyle\frac{1}{\beta + \frac{\sqrt{\lambda_n}}{\pi}} \left|G_{\lambda_n} \right\rangle \left\langle G_{\lambda_n}\right| \right]  \psi\right. \right\rangle_2 \\
    =\left\langle \phi\left|G_{\lambda_n}\right. \right\rangle_2 \left.\left\langle G_{\lambda_n}\right|\psi \right\rangle_2 \left| \left\langle v\left| \sigma \nu(\sigma) \left(1 + B_{\sigma,n} \right)^{-1}\right|u\right\rangle_2 - \textstyle\frac{1}{\beta + \frac{\sqrt{\lambda_n}}{\pi}} \right| \\
    \leqslant \frac{C}{\sqrt{|\lambda_n|}} \left\| \phi \right\|_2 \left\| \psi \right\|_2 \left| \left\langle v\left| \sigma \nu(\sigma) \left(1 + B_{\sigma,n} \right)^{-1}\right|u\right\rangle_2 - \textstyle\frac{1}{\beta + \frac{\sqrt{\lambda_n}}{\pi}} \right|,
  \end{multline}
  where we have used that $ \left\| G_{\lambda_n} \right\|_2 \leqslant C |\lambda_n|^{-1/4} $. The estimate of the last factor
  can essentially be done as in \cite[Lemma 1.2.4]{albeverio2005sqm}: expanding around $ \sigma = 0 $, we get
  \begin{multline}
    \label{eqp: G expansion}
    \nu(\sigma) G_{\sigma^2 \lambda_n}(\mathbf{x} - \mathbf{x}^{\prime}) = \left( 1 + \sigma \nu^{\prime}(\sigma \theta_1(\sigma) \right) G_{0}(\mathbf{x} - \mathbf{x}^{\prime}) - \frac{\sigma \sqrt{\lambda_n}}{4\pi} e^{- \sqrt{\lambda_n} \sigma \theta_2(\sigma) |\mathbf{x} - \mathbf{x}^{\prime}|} \\
    = \left( 1 + \beta \sigma + o(\sigma) \right) G_{0}(\mathbf{x} - \mathbf{x}^{\prime}) - \frac{\sigma \sqrt{\lambda_n}}{4\pi} e^{-
      \sqrt{\lambda_n} \sigma \theta_2(\sigma) |\mathbf{x} - \mathbf{x}^{\prime}|},
  \end{multline}
  for some $ 0 \leqslant \theta_1, \theta_2 \leqslant 1 $ and the remainder $ o(\sigma) $ in the first term is uniform in $ n $. If we
  plug the above expansion in the expression of $ B_{\sigma,n} $, we deduce
  \begin{equation}
  	\label{eq: 2.56}
    B_{\sigma,n} = (1 + \beta \sigma) B - \tx\frac{\sigma \sqrt{\lambda_n}}{4\pi} \left|u \right\rangle \left\langle v\right| + \mathcal{O}\left(\sigma^2|\lambda_n| \right) + o(\sigma),
  \end{equation}
  where we have estimated the Hilbert-Schmidt norm (recall that $ {\mathcal{W}} $ is smooth and compactly supported).
  \begin{multline*}
    \left\| B_{\sigma,n} - (1 + \sigma \beta) B + \tx\frac{\sigma \sqrt{\lambda_n}}{4\pi} \left|u \right\rangle \left\langle v\right| \right\|_{\mathrm{HS}}^2 \\
    \leqslant C \sigma^2 \left| \lambda_n \right| \int_{\mathbb{R}^6} \mathrm{d} \mathbf{x} \mathrm{d} \mathbf{x}^{\prime} \left| 1 - (1 + \beta \sigma)
      e^{-\sqrt{\lambda_n} \sigma \theta_2(\sigma) |\mathbf{x} - \mathbf{x}^{\prime}|} \right|^2 \left| \mathcal{W}(\mathbf{x}) \right| \left|
      \mathcal{W}(\mathbf{x}^{\prime}) \right| \leqslant C \sigma^{4} \lf(|\lambda_n|^{2} + \lf| \lambda_n \ri| \ri).
  \end{multline*}
  We now reproduce the same estimates as in  \cite[eqs. (1.2.45) -- (1.2.47)]{albeverio2005sqm}: setting $B_1 :=\beta B_0  - (4\pi)^{-1}\sqrt{\lambda_n} \,|u\rangle\, \langle v|$ and denoting by $ \epsilon_{\sigma} $ the error $ O(\sigma^2 \lambda_n) + o(\sigma)$ in \eqref{eq: 2.56}, we get
\bmln{
	\sigma \lf(1 + B_{\sigma,n} \ri)^{-1} = \sigma \lf(1 + B_0+\sigma B_1 + \epsilon_{\sigma} \ri)^{-1} \\
	 = \lf(1 + \sigma(1+ \sigma + B_0)^{-1}(B_1 - 1 ) \ri)^{-1} \sigma \lf(1+ \sigma + B_0\ri)^{-1} + \epsilon_{\sigma}.
}
Let $\phi$ be the resonant function appearing in \cref{defi:resonant} and set $P:= - |\tilde{\phi} \rangle \bra{\phi} $, where $\tilde{\phi} :=\mathrm{sgn}(\mathcal{W})\phi$, we have the expansion
\bdm
	\sigma(1+ \sigma + B_0)^{-1} = P + \mathcal{O}(\sigma),
\edm
which plugged in the previous expression yields
\bmln{
	\sigma \lf(1 + B_{\sigma,n} \ri)^{-1} -\epsilon_{\sigma} = \lf(1 + (P + \mathcal{O}(\sigma) ) (B_1 - 1 )\ri)^{-1} \lf(P + \mathcal{O}(\sigma) \ri) \\
	= \lf(1 + P (B_1 - 1 ) + \mathcal{O}(\sigma) \ri)^{-1} P + \mathcal{O}(\sigma) 	\\
	= \left(1 + \left( \beta + \tx\frac{\sqrt{\lambda_n}}{4\pi} |\langle v\,|\, \phi\rangle|^2 \right)^{-1} (1-\beta) |\phi\rangle \,\langle \tilde{\phi}|  + \tx\frac{\sqrt{\lambda_n}}{4\pi} \left( \beta + \tx\frac{\sqrt{\lambda_n}}{4\pi} |\langle v\,|\, \phi\rangle|^2  \right)^{-1}  \langle \phi  | v\rangle  \, |\phi\rangle  \langle v|  \right)P  + \mathcal{O}(\sigma)  \\
	=-\left( \beta + \tx\frac{\sqrt{\lambda_n}}{4\pi} |\langle v\,|\, \phi\rangle|^2  \right)^{-1} P + \mathcal{O}(\sigma),
}
or, equivalently,
\begin{equation}
\sigma (1 + B_{\sigma,n})^{-1} = -\left( \beta + \frac{\sqrt{\lambda_n}}{4\pi} |\langle v\,|\, \phi\rangle|^2  \right)^{-1} P + \mathcal{O}(\sigma) + \mathcal{O}(\sigma^2\lambda_n )+ o(\sigma).
\end{equation}
Therefore, we have
\begin{equation}
|(3)| = \mathcal{O}(\sigma^2|\lambda_n|^{1/2}) + \mathcal{O}(\sigma |\lambda_n|^{-1/2}) + o(\sigma) |\lambda_n|^{-1/2}
\end{equation}
which, imposing the condition (since we are going to choose $ k = n $ later, this would result in the condition $ \sigma n \ll 1 $ in the statement)
  \begin{equation}
    \sigma \sqrt{|\lambda_n|} = C \sigma \sqrt{k n} \ll 1,
  \end{equation}
becomes
\begin{equation}\label{eqp: (3)bound}
|(3)| =  \mathcal{O}(\sigma^2|\lambda_n|^{1/2})  + o(\sigma).
\end{equation}
Note that this implies that
\begin{equation}\label{eqp: norm Bs}
\lf\|\sigma \nu(\sigma) (1 + B_{\sigma,n})^{-1}\ri\| \leq C |\lambda_n|^{-1/2}.
\end{equation}

  Let us now consider the terms (1) and (2) and, since the argument is the same, let us focus on (1). A
  direct inspection of \cite[Lemma 1.2.2]{albeverio2005sqm} reveals that the Hilbert-Schmidt convergence
  of $ A_{\sigma,n} $ to $ A_n $ is in fact uniform in $ n $ and $ k $. More precisely, the dependence on $ \lambda_n
  $ can be easily scaled out: by setting $ \mathbf{y} = \sqrt{|\lambda_n|} \mathbf{x} $ and $ \mathbf{y}^{\prime} =
  \sqrt{|\lambda_n|} \mathbf{x}^{\prime} $, one has
  \begin{multline}
    |\lambda_n|^{-3/4} \left[ (A_{\sigma,n} - A) \psi \right] (|\lambda_n|^{-1/2} \mathbf{y})  = |\lambda_n|^{-1} \int_{\mathbb{R}^3} \mathrm{d} \mathbf{y}^{\prime} \left[ G_{i}(\mathbf{y}) - G_{i}(\mathbf{y} - \sigma \mathbf{y}^{\prime}) \right] v(|\lambda_n|^{-1/2} \mathbf{y}^{\prime}) \tilde\psi(\mathbf{y}^{\prime})	\\
    = : |\lambda_n|^{-1} \left( D_{\sigma,n} \tilde\psi \right)(\mathbf{y}),
  \end{multline}
  where we have set
  \begin{displaymath}
    \tilde\psi(\mathbf{y}) = |\lambda_n|^{-3/4} \psi(|\lambda_n|^{-1/2} \mathbf{y}),		
  \end{displaymath}
  so that the $ L^2 $ norms of the function is preserved.  Exploiting then the smoothness and boundedness
  of $ v $, one can apply the argument of \cite[Proof of Lemma 1.2.2]{albeverio2005sqm} to show that
  \begin{equation}
    \lim_{\sigma \to 0} \left\| D_{\sigma,n} \right\| = 0, \qquad \mbox{uniformly in } n \in \mathbb{N}.
  \end{equation} 
  Moreover, by the very same scaling argument, one can also easily show that
  \begin{equation}
    \left\| A_{\sigma,n} \right\| \leqslant \frac{C}{|\lambda_n|}, \qquad \left\| C_{\sigma,n} \right\| \leqslant \frac{C}{|\lambda_n|}.
  \end{equation}
  Therefore, using \eqref{eqp: norm Bs} and the above estimates, we obtain
  \begin{equation}
    |(1)| = k^{-5/2} n^{-5/2} o_{\sigma}(1),
  \end{equation}
  which combined with \eqref{eqp: (3)bound} yields \eqref{eqp: delta nks}.
		
  Hence, if $ \tau \sim \frac{1}{n} $, \eqref{eqp: approx 3} becomes
  \begin{displaymath}
    \left\| \left[ e^{- i\mathcal{H}_{\beta} \tau} - e^{-i \mathcal{K}_{\sigma, \beta} \tau} \right] \psi \right\|_2 \leqslant C \left( k^{5/2} n^{3/2}  \sigma^2+ k^{-1/2} n^{-3/2} o_{\sigma}(1) \right)  + \lvert n^{-1}  \rvert_{}^{} o_k(1).
  \end{displaymath}
  If we finally choose $ k = n $ for simplicity, we deduce
  \begin{displaymath}
    \left\| \left[ e^{- i\mathcal{H}_{\beta} \tau} - e^{-i \mathcal{K}_{\sigma, \beta} \tau} \right] \psi \right\|_2 \leqslant C n^{4}  \sigma^{2} + n^{-2} o_{\sigma}(1)   +  o_n(n^{-1}),
  \end{displaymath}
  and plugging this in \eqref{eq: strong weak}, we get for any $ \phi, \psi \in \mathcal{D} $ normalized,
  \begin{displaymath}
    \left| \left\langle \phi\left|\left( \mathcal{V}_{n}(t,s) - \mathcal{V}_{n,\sigma}(t,s) \right) \psi\right. \right\rangle \right| \leqslant C n^{5}  \sigma^{2} + n^{-1} o_{\sigma}(1) + o_n(1).
  \end{displaymath}
\end{proof}
		
Finally, we address the approximation {\it \`{a} la} Yoshida for the two-parameter unitary group $ \mathcal{U}_{\sigma}(t,s) $ generated by the approximant operators $ \mathcal{K}_{\beta(t),\sigma} $ (recall the definition \eqref{eq:happ} of the operator $ \mathcal{K}_{\beta,\sigma}  $). Before addressing the main question, it is
useful to state one more technical Lemma.

\begin{lemma}
  	\label{lem: strong est}
  	\mbox{}	\\
  	Let $ \beta \in C(\R) $ and $ \beta' \in \mathbb{R} $ finite, with $ \beta(s) = \beta' $. Let also $ \psi \in \mathcal{D} $ with $ \left\| \psi \right\|_2 \leq 1 $. Then, for any finite $ s < t \in \R $, with $ t - s \gg \sigma $, and for any $ 0 < \epsilon < 1 $, 
  	\begin{equation}
    		\label{eq: strong est}
    		\left\| \left( {\mathcal{U}_{\sigma}(t,s)} - e^{-i \mathcal{K}_{\beta',\sigma} {(t - s)}} \right) \psi \right\|^{{2}}_2 = \OO\big(\sigma^{-2} (t-s)^{8 - \epsilon}\big).
  	\end{equation}
\end{lemma}

\begin{proof}
  	Since at time $ \tau = 0 $ the l.h.s.\ of \eqref{eq: strong est} vanishes, it is sufficient to estimate its time derivative:
  	\begin{multline}
  		\label{eq: proof strong est 1}
    		\partial_{t} \left\| \left( {\mathcal{U}_{\sigma}(t,s)} - e^{-i \mathcal{K}_{\beta',\sigma} {(t - s)}} \right) \psi \right\|_2^2 = - 2 \Re \left[ \partial_t \left\langle  {\mathcal{U}_{\sigma}(t,s)} \psi\left|  e^{-i \mathcal{K}_{\beta',\sigma} {(t - s)}} \psi\right. \right\rangle	\right] \\
    		= 2 \sigma \lf( \beta' - \beta(t) \ri)  \Im \left[ \left\langle  e^{-i \mathcal{K}_{\beta',\sigma} {(t - s)}} \psi\left|  \mathcal{W}_\sigma \lf(e^{-i \mathcal{K}_{\beta',\sigma} {(t - s)}} - {\mathcal{U}_{\sigma}(t,s)} \ri) \psi\right. \right\rangle	\right]  
  	\end{multline}
  	where we used that the expectation of $ \mathcal{W}_{\sigma} $ is real. Hence, exploiting the differentiability of $ \beta(t) $ to bound $  \lf| \beta(t) - \beta' \ri| \leq C (t - s) $, we get
  	\bmln{
  		\partial_{t} \left\| \left( {\mathcal{U}_{\sigma}(t,s)} - e^{-i \mathcal{K}_{\beta',\sigma} {(t - s)}} \right) \psi \right\|_2^2 \leq C \sigma (t-s)  \lf\|  \mathcal{W}_{\sigma}  e^{-i \mathcal{K}_{\beta',\sigma} {(t - s)}} \psi \ri\|_2 \left\| \left( {\mathcal{U}_{\sigma}(t,s)} - e^{-i \mathcal{K}_{\beta',\sigma} {(t - s)}} \right) \psi \right\|_2.
	}
	Now, we claim that, for any $  \epsilon > 0 $ for any $ s,t \in \R $ such that $ t - s \gg \sigma $,
	\beq
		\label{eq: Wsigma bound}
		\lf\|  \mathcal{W}_{\sigma}  e^{-i \mathcal{K}_{\beta',\sigma} {(t - s)}} \psi \ri\|_2 = \OO\big(\sigma^{-2} (t-s)^{2 - \epsilon} \big),
	\eeq
	so that we obtain the result, {\it i.e.},
	\bml{
		\left\| \left( {\mathcal{U}_{\sigma}(t,s)} - e^{-i \mathcal{K}_{\beta',\sigma} {(t - s)}} \right) \psi \right\|_2 \leq C \sigma (t-s)  \int_s^t \diff \tau \lf\|  \mathcal{W}_{\sigma}  e^{-i \mathcal{K}_{\beta',\sigma} {(\tau - s)}} \psi \ri\|_2 \\
		\leq  C \sigma^{-1} (t - s)^{4 - \epsilon}.
	}
	
	Let us now prove \eqref{eq: Wsigma bound}. We first observe that, for any $ \psi \in \mathcal{D} $ and for some $ a > 5/2 $,
	\beqn
		\label{eq: function 0}
		\lf\| \mathcal{W}_{\sigma}  \psi \ri\|_2^2 & = & \frac{1}{\sigma} \int_{\R^3} \diff \xv \: \mathcal{W}^2(\xv) \lf| \psi(\sigma x) \ri|^2 = \OO(\sigma^{2a - 1}),	\\
		\lf\| \mathcal{W}^2_{\sigma}  \psi \ri\|_2^2 & = & \frac{1}{\sigma^5} \int_{\R^3} \diff \xv \: \mathcal{W}^4(\xv) \lf| \psi(\sigma x) \ri|^2 = \OO(\sigma^{2a - 5}),
	\eeqn
	since $ W $ is compactly supported and by \eqref{eq: domainD}. Furthermore,
	\beq
		\label{eq: derivative}
		\partial_t \lf\| \mathcal{W}_{\sigma}  e^{-i \mathcal{K}_{\beta',\sigma} {(t - s)}} \psi \ri\|_2^2 = 2 \Im \meanlrlr{\mathcal{K}_{\beta',\sigma} \psi(t)}{\mathcal{W}_{\sigma}^2}{\psi(t)},
	\eeq
	where we abbreviated $ \psi(t) : = e^{-i \mathcal{K}_{\beta',\sigma} (t - s)} \psi $. Hence,
	\beq
		\label{eq: derivative 0}
		\lf. \partial_t \lf\| \mathcal{W}_{\sigma}  \psi(t) \ri\|_2^2 \ri|_{t = s}  \leq C \lf\| \mathcal{K}_{\beta',\sigma} \psi \ri\|_2 \lf\| \mathcal{W}_{\sigma}^2 \psi \ri\|_2 \leq  C \lf( 1 + \sigma^{a - \frac{1}{2}} \ri) \sigma^{a - \frac{5}{2}} = \OO\big(\sigma^{a - \frac{5}{2}}\big).
	\eeq
	Finally, we compute now and estimate the second derivative of the quantity in \eqref{eq: Wsigma bound}:
	\bml{
		\label{eq: second derivative}
		 \partial^2_t \lf\| \mathcal{W}_{\sigma}  \psi(t) \ri\|_2^2 = 2 \Re \lf\{ \meanlrlr{\mathcal{K}_{\beta',\sigma}^2 \psi(t)}{\mathcal{W}_{\sigma}^2}{\psi(t)} -  \meanlrlr{\mathcal{K}_{\beta',\sigma}\psi(t)}{\mathcal{W}_{\sigma}^2}{\mathcal{K}_{\beta',\sigma} \psi(t)} \ri\} \\
		\leq 2 \lf\| \mathcal{K}_{\beta',\sigma}^2 \psi(t) \ri\|_2 \lf\| \mathcal{W}_{\sigma}^2 \psi(t) \ri\|_2 \leq \frac{C}{\sigma^2} \lf\| \mathcal{K}_{\beta',\sigma}^2 \psi(t) \ri\|_2 \lf\| \mathcal{W}_{\sigma} \psi(t) \ri\|_2 \leq \frac{C}{\sigma^2} \lf\| \mathcal{W}_{\sigma} \psi(t) \ri\|_2 ,
	}
	where we have bounded the $ L^{\infty} $ norm of $ \mathcal{W}_{\sigma} $ by $ C/\sigma^2 $ and estimated
	\beq
		\lf\| \mathcal{K}_{\beta',\sigma}^2 \psi(t) \ri\|_2 = \lf\| \mathcal{K}_{\beta',\sigma}^2 \psi(s) \ri\|_2\leq C \lf[ \big\| \lf( -\Delta \ri)^2 \psi \big\|_2	+ \lf\| \mathcal{W}_{\sigma}^2 \psi \ri\|_2 \ri] = \OO(1) + \OO\big( \sigma^{a - \frac{5}{2}} \big).
	\eeq
	
	The idea is to start with a preliminary bound on $  \lf\| \mathcal{W}_{\sigma} \psi(t) \ri\|_2 $ and then refine recursively through \eqref{eq: function 0}, \eqref{eq: derivative 0} and \eqref{eq: second derivative}. The starting point is the bound
	\beq
		\label{eq: starting estimate}
		\lf\| \mathcal{W}_{\sigma}  e^{-i \mathcal{K}_{\beta',\sigma} {(t - s)}} \psi \ri\|_2 = \OO(\sigma^{-2} (t-s)),
	\eeq
	which can be proven as follows:  using the $L^{\infty}$ bound on $ \mathcal{W}_{\sigma} $ in \eqref{eq: derivative}, we get
	\beq
		\partial_t \lf\| \mathcal{W}_{\sigma}  \psi(t) \ri\|_2^2 \leq \frac{C}{\sigma^{2}} \lf\| \mathcal{K}_{\beta',\sigma} \psi(t) \ri\|_2 \lf\| \mathcal{W}_{\sigma}  \psi(t) \ri\|_2,
	\eeq
	yielding $ \partial_t \lf\| \mathcal{W}_{\sigma}  \psi(t) \ri\|_2 \leq C \sigma^{-2} $, which implies \eqref{eq: starting estimate} via \eqref{eq: function 0} and the condition $ t -s \gg \sigma $. Plugging \eqref{eq: starting estimate} into \eqref{eq: second derivative} and combining it with \eqref{eq: function 0} and \eqref{eq: derivative 0}, we find
	\beq
		\lf\| \mathcal{W}_{\sigma}  \psi(t) \ri\|^2_2 = \mathcal{O}(\sigma^{2a  -1}) + \mathcal{O}\big( \sigma^{a - \frac{5}{2}} (t-s) \big) + \mathcal{O}\big( \sigma^{-4} (t-s)^3 \big) = \mathcal{O}\big( \sigma^{-4} (t-s)^3 \big).
	\eeq
	If now we plug the above in place of \eqref{eq: starting estimate} into \eqref{eq: second derivative}, we end up with the bound $ \lf\| \mathcal{W}_{\sigma}  \psi(t) \ri\|^2_2 = \mathcal{O}\big( \sigma^{-4} (t-s)^{7/2} \big) $, {\it i.e.}, at each step we improve the dependence on $ t -s $. After a direct check, one realizes that after $ N \in \N $ steps of the bootstrap, the exponent of $(t-s) $ in the estimate above reads $ 2 + \sum_{k = 0}^{N-1} \frac{1}{2^k} $, which yields \eqref{eq: Wsigma bound}, after suitably large number of repetitions of the argument.
\end{proof}
	
We are now in position to prove the last estimate on the Yoshida approximation for the dynamics generated
by $ \mathcal{K}_{\beta(t),\sigma} $.
	
	\begin{proposition}[Step function approximation of $ \mathcal{U}_{\sigma}(t,s) $]
  		\label{lem: approx 2}
  		\mbox{}	\\
  		Let $ \beta  \in C(\R) $ and let $ \psi \in \mathcal{D}  $ with $ \left\| \psi \right\|_2 = 1 $. Then, for any finite $ s < t \in \R $ and for any $ 0 < \epsilon < 1 $, there exists a constant $ C > 0 $ independent of $ t,s $
  		\begin{equation}
    			\label{eq: approx 2}
    			\left\| \left( \mathcal{V}_{\sigma, n}(t,s) -  \mathcal{U}_{\sigma}(t,s) \right) \psi \right\|_2^2 = \OO\big(n^{-6 + \epsilon} \sigma^{-2} \big).
 	 	\end{equation}
	\end{proposition}

\begin{proof}
	By a direct application of \cref{lem: strong est} with $ t = t_1 $, we get
	\beq
		\left\| \left( \mathcal{V}_{\sigma, n}(t_1,s) -  \mathcal{U}_{\sigma}(t_1,s) \right) \psi \right\|^2_2  = \OO\big( n^{-8 + \epsilon} \sigma^{-2}\big)
	\eeq
	by differentiability of $ \beta(t) $. Note the condition $ n \sigma \ll 1 $ inherited from $ t_1-s \gg \sigma $ in the statement of \cref{lem: strong est}. Furthermore, we claim that, for any $ j \in \lf\{1, \ldots, n-1 \ri\} $,
	\beq
		\left\| \left( \mathcal{V}_{\sigma, n}(t_{j+1},s) -  \mathcal{U}_{\sigma}(t_{j+1},s) \right) \psi \right\|_2 \leq \left\| \left( \mathcal{V}_{\sigma, n}(t_{j},s) -  \mathcal{U}_{\sigma}(t_{j},s) \right) \psi \right\|_2 + C n^{-4 + \epsilon/2} \sigma^{-1},
	\eeq
	and the result then follows by a trivial recursive argument. In order to show that the above estimate holds true, we write
	\bmln{
		{\left\| \left(  \mathcal{U}_{\sigma}(t_{j+1},s) - \mathcal{V}_{\sigma, n}(t_{j+1},s) \right) \psi \right\|_2}	\\
		\leq \left\| \left(\mathcal{U}_{\sigma}(t_{j+1},s) -  \mathcal{U}_{\sigma}(t_{j+1},t_j)\mathcal{V}_{n,\sigma}(t_{j},s) \ri) \psi \ri\|_2 + \lf\| \left( \mathcal{U}_{\sigma}(t_{j+1},t_j)\mathcal{V}_{n,\sigma}(t_{j},s) -  \mathcal{V}_{\sigma, n}(t_{j+1},s) \right) \psi \ri\|_2 \\
		\leq  \left\| \left(\mathcal{U}_{\sigma}(t_{j},s) - \mathcal{V}_{\sigma, n}(t_{j},s)  \ri) \psi \ri\|_2 + \lf\| \left( \mathcal{U}_{\sigma}(t_{j+1},t_j) -  \mathcal{V}_{\sigma, n}(t_{j+1},t_j)  \right) \psi(t_j) \ri\|_2	\\
		\leq 	\left\| \left(\mathcal{U}_{\sigma}(t_{j},s) - \mathcal{V}_{\sigma, n}(t_{j},s)  \ri) \psi \ri\|_2 + \lf\| \left( \mathcal{U}_{\sigma}(t_{j+1},t_j) -  \mathcal{V}_{\sigma, n}(t_{j+1},t_j)  \right) \chi_m \ri\|_2 + C m^{-2/3}
	}
	where we have set $ \psi(t_j) : = \mathcal{V}_{n,\sigma}(t_{j},s) \psi $ for short and applied once more \cref{lem: approx smooth}. Now,  we can use \cref{lem: strong est} to bound the second term in the expression above and choose $ m \gg n^{6 - 3\epsilon} \sigma^{3/2} $ to get that the sum of last two terms above is bounded by $ C n^{-4 + \epsilon/2} \sigma^{-1} $. 
\end{proof}
        
We now complete the proof of the main result in this Section.

\begin{proof}[Proof of \cref{pro:approx}]
  The idea is to prove the result in three steps: we first replace $ \mathcal{U}_{\mathrm{eff}}(t,s) $
  with its Yoshida approximants; the resulting dynamics is then generated by a time-independent point
  interaction and, as such, can be approximated by $ \mathcal{V}_{n,\sigma}(t,s) $, i.e., the dynamics
  generated by $ \mathcal{K}_{\beta_j,\sigma} $ in the corresponding interval; finally, we undo the step function
  approximation of $ \beta(t) $ and obtain $ \mathcal{U}_{\sigma}(t,s) $.
          
  To prove strong convergence of unitary operators, it is sufficient to prove weak
    convergence on a dense subset. This can be easily proved as follows. Let
    $V_n \xrightarrow[n\to \infty]{\mathrm{w}} V$ on a dense subset $\mathscr{D}\subset \mathscr{H}$. In addition, given
    $\psi\in \mathscr{H}$, let us denote by $ \lf\{ \psi_m \ri\}_{m\in \mathbb{N}}$ its approximation in $\mathscr{D}$. Then
    \begin{equation*}
      \begin{split}
        \textstyle\frac{1}{2} \left\|  (V-V_n)\psi  \right\|_{}^2 \leqslant \left\| (V-V_n)\psi_m  \right\|_{}^2 + \left\| (V-V_n)(\psi-\psi_m)  \right\|_{}^2\\\leqslant \left\langle \psi_m\left|\left(2-V^{*}V_n-V^{*}_n V \right)\psi_m\right. \right\rangle_2 + 4 \left\| \psi-\psi_m \right\|_{2}^2\\
        \leqslant 2\Re \left\langle  V\psi_m\left| \left(V-V_n \right)\psi_m\right. \right\rangle_2+4 \left\|  \psi-\psi_m \right\|_{2}^2
      \end{split}
    \end{equation*}
    that converges to zero as $n\to \infty$, since $m$ can be chosen arbitrarily large. Hence, it is sufficient
    to prove the convergence for all $\psi,\phi\in \mathcal{D} $.

  Now, for any $ \psi,\phi\in \mathcal{D} $ and for all $ t,s \in \mathbb{R} $ with $ s < t $:
  \begin{multline}
    \left| \left\langle \phi\left|\left( \mathcal{U}_{\mathrm{eff}}(t,s) - \mathcal{U}_{\sigma}(t,s) \right) \psi\right. \right\rangle \right| \leqslant \left| \left\langle \phi\left|\left( \mathcal{U}_{\mathrm{eff}}(t,s) - \mathcal{V}_{n}(t,s) \right) \psi\right. \right\rangle \right|	\\
    + \left| \left\langle \phi\left|\left( \mathcal{V}_{n}(t,s) - \mathcal{V}_{n,\sigma}(t,s) \right)\psi\right. \right\rangle \right| + \left| \left\langle \phi\left|\left( \mathcal{V}_{n,\sigma}(t,s) - \mathcal{U}_{\sigma}(t,s) \right) \psi\right. \right\rangle \right|.
  \end{multline}
  Hence, using the results proven in \cref{lem: approx 1,lem: approx 3,lem: approx
    2}, respectively, we get
  \beq
    \left| \left\langle \phi\left|\left( \mathcal{U}_{\mathrm{eff}}(t,s) - \mathcal{U}_{\sigma}(t,s) \right) \psi\right. \right\rangle \right|  = \mathcal{O}(n^5 \sigma^{2}) + n^{-1} o_{\sigma}(1) + o_n(1) + \OO\big(n^{-3+\epsilon/2} \sigma^{-1} \big)
  \eeq
  provided that $\sigma n \ll 1$. If we now optimize the first and last terms, {\it i.e.}, we pick
  \begin{equation}
    n = n_{\sigma} = \sigma^{-3/(8 - \epsilon/2)} \xrightarrow[\sigma \to 0]{} + \infty,
  \end{equation}
  for $ \epsilon > 0 $ small enough the condition $ n \sigma \ll 1 $ is satisfied, and
  \begin{displaymath}
    \left| \left\langle \phi\left|\left( \mathcal{U}_{\mathrm{eff}}(t,s) - \mathcal{U}_{\sigma}(t,s) \right) \psi\right. \right\rangle \right| = \mathcal{O}\big(\sigma^{1/8 - C \epsilon} \big)  + o_{\sigma}(1) \xrightarrow[\sigma \to 0]{} 0.
  \end{displaymath}
\end{proof}
        
\section{Convergence of Fluctuations}
\label{sec:conv-fluct}

In this Sect.\ we prove the quasi-classical convergence, in strong topology, of the unitary operator of
microscopic coherent quantum fluctuations, perturbing the quasi-classical solution. The key idea is to use
coherent states that, in the quasi-classical limit, are ``singular enough'' to produce an effective point
interaction. The strong convergence of fluctuations is sufficient to prove the strong convergence of
evolved particle observables given in \cref{thm:eff}. Let us start with some preliminary definitions and
remarks.

In order to have a more compact notation for the two phonon fields, we introduce a single boson field
  encompassing both. As it is well known, $\Gamma_{\mathrm{sym}}(\mathfrak{H})\otimes \Gamma_{\mathrm{sym}}(\mathfrak{H})\cong \Gamma_{\mathrm{sym}}(\mathfrak{H}\oplus
  \mathfrak{H})$. With this identification, we can introduce vector creation and annihilation operators
\begin{equation}
  \label{eq:3}
  \mathbf{a}_{\varepsilon}^{\sharp}(\bm{\eta})=a_{\varepsilon}^{\sharp}(\eta_1)+b_{\varepsilon}^{\sharp}(\eta_2)\; , \quad \forall \bm{\eta}=(\eta_1,\eta_2)\in \mathfrak{H}\oplus \mathfrak{H}\; ;
\end{equation}
and analogously the second quantization
\begin{equation}
  \label{eq:4}
  \mathrm{d}\bm{\Gamma}_{\varepsilon}(\mathbf{h})=\mathrm{d}\Gamma_{\varepsilon}^{(a)}(h_1)+\mathrm{d}\Gamma_{\varepsilon}^{(b)}(h_2)\; ,
\end{equation}
where $h_1,h_2$ are self-adjoint operators on $\mathfrak{H}$. Hence it follows that, defining
  $\bm{\lambda}_{\mathbf{x}}=(\lambda_{\mathbf{x}}^{(a)},\lambda_{\mathbf{x}}^{(b)})$, and
  $\bm{\omega}_{\varepsilon}=(\omega,\frac{\kappa}{\varepsilon})$, the Hamiltonian $H_{\varepsilon}$ can be rewritten in the compact form:
\begin{equation}
  \label{eq:7}
  H_{\varepsilon}=-\Delta+\mathrm{d}\bm{\Gamma}_{\varepsilon}(\bm{\omega}_{\varepsilon})+\mathbf{a}_{\varepsilon}(\bm{\lambda}_{\mathbf{x}})+\mathbf{a}_{\varepsilon}^{\dagger}(\bm{\lambda}_{\mathbf{x}})\; .
\end{equation}
Let us remark that the form factor
  $\bm{\lambda}_{\mathbf{x}}(\mathbf{k})=e^{i \mathbf{k}\cdot\mathbf{x}}(\lambda_0(\mathbf{k}),k^{-1})$ has
  the following important property, as first remarked in \cite{lieb1997cmp}:
\begin{align}\label{lambdacond}
  &\bm{\lambda}_\mathbf{x}  = \bm{\lambda}_{>,\mathbf{x}} + \bm{\lambda}_{<,\mathbf{x}}\;; \\
  & \bm{\lambda}_{<,\mathbf{x}} \in L^{\infty}(\mathbb{R}^3; \mathfrak{H}\oplus \mathfrak{H})\;; \\
  & \bm{\lambda}_{>,\mathbf{x}} = [-i\nabla, \bm{\xi}_\mathbf{x}], \qquad \bm{\xi}_\mathbf{x} \in L^{\infty}(\mathbb{R}^{3}; (\mathfrak{H}\oplus \mathfrak{H})^3)\;.
\end{align}
{More precisely, for all $r> 0$, it is possible to make the splitting
  $\bm{\lambda}_{<}=(\lambda_{1}, \lambda_{2,<})$, $\bm{\xi}=(0,\xi_{2,r})$ in a way such that
  $\lVert \xi_{2,r}  \rVert_{L^{\infty}(\mathbb{R}^{3}; (\mathfrak{H})^3)}^{}\xrightarrow[r\to
  +\infty]{}0$  and $\lVert \lambda_{2,<}  \rVert_{L^{\infty}(\mathbb{R}^{3}; \mathfrak{H})}^{}\xrightarrow[r\to
  +\infty]{} +\infty$.}

{Throughout the rest of the paper, whenever a Fock space estimate is
  used, it will be a direct application of the following basic estimate,
  whose well-known proof stems combining a direct calculation and the
  canonical commutation relations: for any couple of positive self-adjoint
  operators $\bm{\tau}=(\tau_1,\tau_2)$ on $ \mathfrak{H} \oplus \mathfrak{H}$, for any $\mathbf{f},\mathbf{g}$ such that
  $\mathbf{f},\mathbf{g},\bm{\tau}^{-\frac{1}{2}}\mathbf{f},\bm{\tau}^{-\frac{1}{2}}\mathbf{g}\in
  L^{\infty}\bigl(\mathbb{R}^3; (\mathfrak{H}\oplus \mathfrak{H})^3\bigr)$, and for any $0<\varepsilon\leq 1$,}
\begin{multline}
  \label{eq:2}
  {\lf\| \lVert \lf( \mathbf{a}_{\varepsilon}(\mathbf{f})+\mathbf{a}^{\dagger}_{\varepsilon}(\mathbf{g}) \ri) \lf( \mathrm{d}\Gamma_{\varepsilon}(\bm{\tau})+1 \ri)^{-\frac{1}{2}}  \ri\rVert_{}^{}\leq C \lf( \lf\| \bm{\tau}^{-\frac{1}{2}}\mathbf{f} \ri\|_{L^{\infty}(\mathbb{R}^3; (\mathfrak{H}\oplus \mathfrak{H})^3)}^{}+ \lf\| \bm{\tau}^{-\frac{1}{2}}\mathbf{g}  \ri\|_{L^{\infty}(\mathbb{R}^3; (\mathfrak{H}\oplus \mathfrak{H})^3)}^{} \ri.} \\
  {\lf. +\varepsilon \lf\|\mathbf{g}  \ri\|_{L^{\infty}(\mathbb{R}^3; (\mathfrak{H}\oplus \mathfrak{H})^3)}^{} \ri)\; ,}
\end{multline}
{where the norm on the l.h.s.\ is the operator norm on $\mathscr{H}$. As
  a straightforward application, making use of \eqref{lambdacond} this
  inequality yields that $H_{\varepsilon}$ is a densely defined quadratic form on
  $\mathscr{D}[H_0]$: there exist $A>0$ such that for all $\Theta\in
  \mathscr{D}[H_0]$, and for all $\delta>0$, there exists $B_{\delta}>0$ such that,
  uniformly w.r.t. $\varepsilon\in (0,1)$,}
\begin{equation}
  \label{eq:13}
  {\meanlrlr{\Theta}{H_{\varepsilon}}{\Theta}_{\mathscr{H}}\leq A \lf(\delta \lf\| \bm{\lambda}_{<} \ri\|_{L^{\infty}(\mathbb{R}^3; \mathfrak{H}\oplus \mathfrak{H})}^{} + \lf\| \xi_{2,r} \ri\|_{L^{\infty}(\mathbb{R}^3; (\mathfrak{H})^3)}^{} \ri) \meanlrlr{\Theta}{H_0}{\Theta}_{\mathscr{H}}+ B_{\delta} \lf\|  \Theta  \ri\|_{\mathscr{H}}^2\; ,}
\end{equation}
{where $B_{\delta}$ depends additionally on $\lVert \bm{\lambda}_{<}
  \rVert_{L^{\infty}(\mathbb{R}^3; \mathfrak{H}\oplus \mathfrak{H})}^{}$ and $\lVert \xi_{2,r}
  \rVert_{L^{\infty}(\mathbb{R}^3; \mathfrak{H}^3)}^{}$. Furthermore, it is possible to choose
  $\delta>0$ and $r>0$ in the above bound such that}
\begin{equation}
  \label{eq:14}
  {A \lf(\delta \lf\| \bm{\lambda}_{<} \ri\|_{L^{\infty}(\mathbb{R}^3; \mathfrak{H}\oplus \mathfrak{H})}^{} + \lf\| \xi_{2,r} \ri\|_{L^{\infty}(\mathbb{R}^3; (\mathfrak{H})^3)}^{} \ri)<1\; ,}
\end{equation}
{so that the quadratic form induced by $H_{\varepsilon}$ is bounded from below and
  symmetric by KLMN's theorem. As already remarked, the above argument is
  completely analogous to the one commonly given to prove self-adjointness of
  the optical polaron model (see, \emph{e.g.}, \cite{frank2014lmp}).}

The operator of fluctuations for coherent states is now defined as follows. Let
  $\mathbf{W}_{\varepsilon}(\frac{\bm{\alpha}_{\varepsilon}(t)}{i\varepsilon})=W^{(a)}_{\varepsilon}(\frac{\alpha_{\varepsilon}}{i\varepsilon})W_{\varepsilon}^{(b)}(\frac{\beta_{\varepsilon}}{i\varepsilon})$
be the Weyl operator appearing in the definition of $\Xi_{\varepsilon}$ (recall \eqref{eq:coherent}). Then, the
operator of \emph{microscopic fluctuations} $Z_{\varepsilon}(t,s)$ is defined by
\begin{equation}
  \label{eq:8}
  Z_{\varepsilon}(t,s):=\mathbf{W}_{\varepsilon}^{\dagger}\left(\tfrac{\bm{\alpha}_{\varepsilon}(t)}{i\varepsilon}\right)e^{-iH_{\varepsilon}(t-s)}\mathbf{W}_{\varepsilon}\left(\tfrac{\bm{\alpha}_{\varepsilon}}{i\varepsilon}\right)\; ,
\end{equation}
where $ \bm{\alpha}_{\varepsilon}(t)=(\alpha_{\varepsilon}(t),\beta_{\varepsilon}(t)) $ satisfies the classical dynamics \eqref{eq: alpha evolution}, {\it i.e.}, recalling \eqref{eq:alpha},
\beq
\bm{\alpha}_{\varepsilon}(\kv; t) : = \lf(\alpha_{\varepsilon}(\kv), e^{-i\kappa(t-s)}\beta_{\varepsilon}(\kv) \ri).
\eeq
The strong limit of $ Z_{\eps}(t,s) $ as $\varepsilon\to 0$, of
which we prove the existence, is the operator of \emph{quasi-classical fluctuations} $Z(t,s)$, defined by
\begin{equation}
  \label{eq:9}
  Z(t,s):= \mathcal{U}_{\mathrm{eff}}(t,s)\otimes e^{-i(t-s) \mathrm{d}\bm{\Gamma}((0,\kappa))}= \mathcal{U}_{\mathrm{eff}}(t,s)\otimes 1\otimes e^{-i \kappa(t-s)\mathrm{d}\Gamma^{(b)}(1)} \;,
\end{equation}
where $ \mathrm{d}\bm{\Gamma},\mathrm{d} \Gamma^{(b)} $ stand for the second-quantized operators defined in
terms of the unscaled creation and annihilation operators $ a^{\sharp},b^{\sharp} $ (recall
\eqref{eq:dgamma}). Therefore, $ Z $ is a factorized unitary operator on the full space $\mathscr{H}$, and
its factorization is due to the chosen scaling in $H_{\varepsilon}$, that guarantees no quasi-classical
back-reaction on the field.

The relation between the fluctuation operators $ Z_{\varepsilon}(t,s), Z(t,s) $ and the full Heisenberg evolution of a particle observable $ \mathcal{B} $ can be derived as follows. Let $\mathcal{B}$ be a bounded particle operator, acting on $L^2(\mathbb{R}^3)$, and let
$\mathcal{B}_{\varepsilon}(t,s)$ and $\mathcal{B}(t,s)$ be the associated microscopic and quasi-classical Heisenberg
evolved operators as defined in \eqref{eq:ev:start} and \eqref{eq:ev:obs}, respectively, then, for any $t,s\in \mathbb{R}$ and any $\psi\in L^2 (\mathbb{R}^3)$,
\begin{multline*} 
    \left\| \bigl(\mathcal{B}(t,s)-\mathcal{B}_{\varepsilon}(t,s)\bigr)\psi  \right\|_{L^2(\mathbb{R}^3)}^2 =  \left\| \bigl(B(t,s)-B_{\varepsilon}(t,s)\bigr)\psi \otimes \Xi_{\varepsilon,s}  \right\|_{\mathscr{H}}^2 \\
    = \left\langle \psi \otimes \Xi_{\varepsilon,s}\left| \left|{B}(t,s) \right|^2 + \left| {B}_{\varepsilon}(t,s)\right|^2 - {B}^*(t,s){B}_{\varepsilon}(t,s) -{B}^*_{\varepsilon}(t,s){B}(t,s) \right| \psi \otimes \Xi_{\varepsilon,s}\right\rangle_{\mathscr{H}}.
 \end{multline*}
If we now plug in the definition of $ Z $ and $ Z_{\varepsilon} $ and use \eqref{eq:coherent}, we can write
\begin{multline*} 
	   \left\| \bigl(\mathcal{B}(t,s)-\mathcal{B}_{\varepsilon}(t,s)\bigr)\psi  \right\|_{L^2(\mathbb{R}^3)}^2 = \left\langle \psi \otimes \Omega\left| \mathbf{W}_{\varepsilon}^{\dagger} \left( {B}(t,s) - {B}_{\varepsilon}(t,s) \right)^2 \mathbf{W}_{\varepsilon} \right| \psi \otimes \Omega\right\rangle_{\mathscr{H}} \\
	   = \left\langle \psi \otimes \Omega\left| Z_{\varepsilon}^\dagger \mathcal{B}^2 \otimes 1 Z_{\varepsilon} + {Z}_{\varepsilon}^\dagger \mathcal{B}^2 \otimes 1 {Z}_{\varepsilon} - {Z}_{\varepsilon}^{\dagger} \mathcal{B} \otimes 1 {Z}_{\varepsilon} Z_{\varepsilon}^{\dagger} \mathcal{B} \otimes 1 Z_{\varepsilon} - Z_{\varepsilon}^{\dagger} \mathcal{B} \otimes 1 Z_{\varepsilon} {Z}_{\varepsilon}^{\dagger} \mathcal{B} \otimes 1 {Z}_{\varepsilon}\right| \psi \otimes \Omega\right\rangle_{\mathscr{H}}	
 \end{multline*}	   
where we used the following property of Weyl operators (recall \eqref{eq: alpha evolution})
\begin{equation}
	e^{-i (t - s)\mathrm{d} \bm{\Gamma}(\mathbf{h}) } \mathbf{W}_{\varepsilon}\left( \tfrac{\bm{\alpha}_{\varepsilon}}{i \varepsilon} \right) = \mathbf{W}_{\varepsilon}\left( \tfrac{e^{-i(t-s) \mathbf{h}} \bm{\alpha}_{\varepsilon}}{i \varepsilon} \right)  e^{-i(t - s) \mathrm{d} \bm{\Gamma}(\mathbf{h}) } = \mathbf{W}_{\varepsilon}\left( \tfrac{\bm{\alpha}_{\varepsilon}(t)}{i \varepsilon} \right)  e^{-i(t-s) \mathrm{d} \bm{\Gamma}(\mathbf{h})}\; ,
\end{equation}
which, combined with \eqref{eq:dgamma}, {\it i.e.}, $ \varepsilon^{-1} \mathrm{d} \bm{\Gamma}_{\varepsilon}(\mathbf{h}) =
\mathrm{d} \bm{\Gamma}(\mathbf{h}) $, implies, for any particle operator $ \mathcal{A} $,
\begin{displaymath}
	\mathbf{W}_{\varepsilon}^{\dagger}\left(\textstyle\frac{\bm{\alpha}_{\varepsilon}(t)}{i\varepsilon}\right) \left( \mathcal{U}_{\mathrm{eff}}^{\dagger} \otimes 1 \right) \mathcal{A} \otimes 1 \left( \mathcal{U}_{\mathrm{eff}} \otimes 1 \right) \mathbf{W}_{\varepsilon}\left(\textstyle\frac{\bm{\alpha}_{\varepsilon}}{i\varepsilon}\right) = {Z}^{\dagger} \mathcal{A} \otimes 1 {Z},
\end{displaymath}
and we omit the dependence on $t$ and $s$ of $Z(t,s)$ and $Z_{\varepsilon}(t,s)$ for convenience. If we now exploit the identity
\begin{multline*} 
	Z_{\varepsilon}^\dagger \mathcal{B}^2 \otimes 1 Z_{\varepsilon} + {Z}^\dagger \mathcal{B}^2 \otimes 1 {Z} - {Z}^{\dagger} \mathcal{B} \otimes 1 {Z} Z_{\varepsilon}^{\dagger} \mathcal{B} \otimes 1 Z_{\varepsilon} - Z_{\varepsilon}^{\dagger} \mathcal{B} \otimes 1 Z_{\varepsilon} {Z}^{\dagger} \mathcal{B} \otimes 1 {Z} \\
	= (Z^{\dagger}_{\varepsilon} -{Z}^{\dagger})\mathcal{B}^2\otimes 1 Z_{\varepsilon} 
	+ {Z}^{\dagger} \mathcal{B}^2\otimes 1({Z}-Z_{\varepsilon}) - {Z}^{\dagger} \mathcal{B}\otimes 1 {Z} (Z_{\varepsilon}^{\dagger}-{Z}^{\dagger}) \mathcal{B}\otimes 1 Z_{\varepsilon} + \mbox{h.c.}
 \end{multline*}
in the expression above, we deduce that
\begin{equation}
	\label{eq: estimate Z}
	\left\| \bigl(\mathcal{B}(t,s)-\mathcal{B}_{\varepsilon}(t,s)\bigr)\psi  \right\|_{L^2(\mathbb{R}^3)}^2     \leqslant 4 \left\| \mathcal{B}  \right\|^{{2}} \left\| \psi  \right\|_2^{} \left\| \Omega  \right\|_{\Gamma_{s}\otimes \Gamma_{s}}^{} \left\| ({Z}-Z_{\varepsilon})\psi\otimes \Omega  \right\|_{\mathscr{H}}^{}.
\end{equation}

The estimate \eqref{eq: estimate Z} makes apparent the link between the Heisenberg evolution of the observable $ \mathcal{B} $ and the fluctuation operators $ Z_{\varepsilon} $, $ Z $. More precisely, the convergence stated in \cref{thm:eff} is equivalent to show strong convergence of the fluctuation operator $ Z_{\varepsilon} $ to $ Z $, which we are going to prove in next \cref{pro:1}. Note that a similar quasi-classical limit of coherent state fluctuations has been studied (with less singular
  coherent states that do not carry any $\varepsilon$-dependence on the classical solution $\bm{\alpha}$) for the renormalized Nelson model in \cite{ginibre2006ahp}.

\subsection{Strong convergence}

In order to prove strong convergence of $Z_{\varepsilon}$, we make use of an intermediate auxiliary operator
\begin{equation}
  \label{eq:1}
  {Y}_\sigma(t,s)=\mathcal{U}_\sigma(t,s)\otimes e^{-i(t-s)\mathrm{d}\bm{\Gamma}((0,\kappa))}\; ,
\end{equation}
with $ \sigma = \sigma_\varepsilon $ properly chosen, {\it i.e.,} such that
\begin{equation}
	\varepsilon^{1/{j_*}} \ll \sigma \ll 1,
\end{equation} 
and where $\mathcal{U}_\sigma$ is the two-parameter group defined in \cref{sec:approx}, and generated by
$\mathcal{K}_{\beta(t),\sigma}$. The precise
result is given in the following

\begin{proposition}[Convergence of fluctuations]
  \label{pro:1}
  	\mbox{}	\\
  For any $\Phi\in \mathscr{H}$,
  \begin{equation}
    \lim_{\varepsilon\to 0} \, \left\| \bigl(Z(t,s)-Z_{\varepsilon}(t,s)\bigr)\Phi  \right\|_{\mathscr{H}}^{}=0\;.
  \end{equation}
\end{proposition}

Before proving the above result, let us give some preparatory lemmas. The first one is a well-known result
about Weyl operators (for an explicit proof, see, \emph{e.g.}, \cite{falconi2012phd}). For the last
property, it is useful to remark that $ \bm{\alpha}_{\varepsilon}(\mathbf{k}; t) $ is by construction a rapidly
decaying function in Schwartz class and therefore its scalar product by any polynomial, as, \emph{e.g.},
$\omega$, is always bounded.

	\begin{lemma}\label{wderivative}
		\mbox{}	\\
  		The Weyl operators $\mathbf{W}_{\varepsilon}\left(\tfrac{\bm{\alpha}_{\varepsilon}(t)}{i \varepsilon}\right)$ are strongly differentiable with respect to
  $t\in \mathbb{R}$ on $\mathscr{D}\bigl(\mathrm{d}\bm{\Gamma}(\bm{1})^{1/2}_{\varepsilon}\bigr)$, with derivative given by
  		\begin{align}
  			  i \partial_t  \mathbf{W}_{\varepsilon}\left(\tfrac{\bm{\alpha}_{\varepsilon}(t)}{i \varepsilon}\right) & =\frac{i}{\varepsilon} \Bigl(\mathbf{a}_{\varepsilon}^{\dagger}(\dot{\bm{\alpha}}_{\varepsilon}) - \mathbf{a}_{\varepsilon}(\dot{\bm{\alpha}}_{\varepsilon}) - i\, \Im \left\langle \bm{\alpha}_{\varepsilon}\left|\dot{\bm{\alpha}}_{\varepsilon}\right. \right\rangle_{\mathfrak{H}\oplus \mathfrak{H}} \Bigr) \, \mathbf{W}_{\varepsilon}\left(\tfrac{\bm{\alpha}_{\varepsilon}(t)}{i \varepsilon}\right)	\nonumber \\
  			&= \frac{i}{\varepsilon}  \mathbf{W}_{\varepsilon}\left(\tfrac{\bm{\alpha}_{\varepsilon}(t)}{i \varepsilon}\right) \Bigl(\mathbf{a}_{\varepsilon}^{\dagger}(\dot{\bm{\alpha}}_{\varepsilon}) - \mathbf{a}_{\varepsilon}(\dot{\bm{\alpha}}_{\varepsilon}) + i\, \Im \left\langle \bm{\alpha}_{\varepsilon}\left|\dot{\bm{\alpha}}_{\varepsilon}\right. \right\rangle_{\mathfrak{H}\oplus \mathfrak{H}} \Bigr)\; .
 		\end{align}
                In addition, $\mathbf{W}_{\varepsilon}(\zv)$ maps $\mathscr{D}[H_0]$ and
                  $\mathscr{D}\bigl(\mathrm{d}\bm{\Gamma}_{\varepsilon}((h,h'))^{1/2}\bigr)$ into themselves for any
                  self-adjoint and positive $h,h'$, provided that $ \zv \in \dom(h) \oplus \dom\lf(h'\ri) $.
	\end{lemma}

Another useful result is the weak differentiability of $Z_{\varepsilon}$ in a suitable dense domain.

\begin{lemma}
  \label{leps}
  \mbox{}	\\
  The operator $Z_{\varepsilon}(t,s)$ is weakly differentiable, with respect to both $t\in \mathbb{R}$ and $s\in \mathbb{R}$, on
  $\mathscr{D}[H_0]\cap \mathscr{D}(\mathrm{d}\bm{\Gamma}_{\varepsilon}(\mathbf{1}))$. The weak derivatives have the
  following form:
  \begin{equation*}
    \begin{split}
      &i\partial_tZ_{\varepsilon}(t,s)= L_{\varepsilon}(t)Z_{\varepsilon}(t,s)\; ,\\
      &i\partial_sZ_{\varepsilon}(t,s)= -Z_{\varepsilon}(t,s)L_{\varepsilon}(s)\; ,
    \end{split}
  \end{equation*}
  where $\bigl(L_{\varepsilon}(t)\bigr)_{t\in \mathbb{R}}$ is the family of operators
  \begin{equation*}
    L_{\varepsilon}(t) = -\Delta\otimes 1+ 2 \Re \left\langle \bm{\lambda}_{\mathbf{x}}|\bm{\alpha}_\varepsilon(t) \right\rangle_{\mathfrak{H}\oplus \mathfrak{H}} + \mathbf{a}_{\varepsilon}(\bm{\lambda}_{\mathbf{x}} ) + \mathbf{a}_{\varepsilon}^{\dagger}(\bm{\lambda}_{\mathbf{x}} ) + \mathrm{d} \bm{\Gamma}_{\varepsilon}(\bm{\omega}_{\varepsilon})\; .
  \end{equation*}
\end{lemma}
\begin{proof}
  Using \cref{wderivative} and the definition of $Z_{\varepsilon}(t,s)$, it is easy to see that it is possible to
  differentiate with respect to both $t$ and $s$ the quantity
  \begin{equation*}
    \left\langle \Theta\left| Z_{\varepsilon}(t,s)\Phi \right. \right\rangle_{\mathscr{H}}\; ,
  \end{equation*}
  for any $\Theta,\Phi\in \mathscr{D}[H_0]\subset \mathscr{D}\bigl(\mathrm{d}\bm{\Gamma}_{\varepsilon}((\omega,1))^{1/2}\bigr)$. Indeed,
    $e^{-i(t-s)H_{\varepsilon}}$ is weakly differentiable on $\mathscr{D}[H_{\varepsilon}]=\mathscr{D}[H_0]$, and
    $\mathbf{W}_{\varepsilon}(\cdot )$ is strongly differentiable on
    $\mathscr{D}\bigl(\mathrm{d}\bm{\Gamma}_{\varepsilon}(\mathbf{1})^{1/2}\bigr)$, and maps $\mathscr{D}[H_0]$ into
    itself. The explicit form of the derivative is given using again \cref{wderivative}, the action as
  translations of Weyl operators when acting on creation and annihilation operators, and the equation for
  the time derivative of $\bm{\alpha}_{\varepsilon}$, \emph{i.e.}, $i\partial_t\bm{\alpha}_{\varepsilon}(t)= \mathrm{diag}\lf\{ 0, \kappa \ri\} \bm{\alpha}_{\varepsilon}(t)$.
\end{proof}

The two final preparatory results are essentially Gronwall-type estimates for the time-evolved expectation
of the Laplace operator.

	\begin{lemma}\label{laplacecontrol}
		\mbox{}	\\
  		For any $t\in \mathbb{R}$, there exists a finite constant $C_t>0$ such that for every $\Phi \in \mathscr{D}[H_+]=\mathscr{D}[H_0]$,
  		\begin{equation}
    			\left\langle \Phi\left| e^{itH_{\varepsilon}} (-\Delta) e^{-itH_{\varepsilon}}\right|\Phi\right\rangle_{\mathscr{H}} \leqslant C_t \left( \left\langle \Phi\left|H_+  + H_I\right|\Phi\right\rangle_{\mathscr{H}} + \left\| \Phi \right\|^2_{\mathscr{H}} \right)\; ,
  		\end{equation}
  		where $H_+ := -\Delta + \mathrm{d} \bm{\Gamma}_{\varepsilon}((\omega,\kappa)) \geqslant 0$.
	\end{lemma}
	
Before proving the above result, let us point out the more regular behavior as $ \varepsilon \to 0 $ of the operator $
H_+ $ compared with $ H_0 $, since the former has no prefactor $ \varepsilon^{-1} $ in front of the optic phonon's energy.

\begin{proof}
  Since $-\Delta\leqslant H_+$, adding and subtracting $H_I$, it is possible to write for any $\Phi\in \mathscr{D}[H_0]$
  \begin{multline} 
    \label{prop:firstineq}
      \left\langle \Phi\left| e^{itH_{\varepsilon}} (-\Delta) e^{-itH_{\varepsilon}}\right|\Phi\right\rangle_{\mathscr{H}} \leqslant \left\langle \Phi\left| e^{itH_{\varepsilon}} H_+ e^{-itH_{\varepsilon}}\right|\Phi\right\rangle_{\mathscr{H}} = \left\langle \Phi\left|e^{itH_{\varepsilon}} (H_+ +H_I) e^{-itH_{\varepsilon}}\right| \Phi\right\rangle_{\mathscr{H}} \\- \left\langle \Phi\left|e^{itH_{\varepsilon}}  H_I e^{-itH_{\varepsilon}}\right| \Phi \right\rangle_{\mathscr{H}}.
	 \end{multline}
	
  Let us control the two terms of the last sum separately. The form associated to $ H_I$ is a small
  perturbation of the one associated to $H_+$, as it has been proved in other papers dealing with the
  polaron Hamiltonian (see, \emph{e.g.}, \cite{lieb1997cmp,frank2014lmp}). Indeed, there exist $  A \in (0,1)$
  and $B >0$, such that
  \begin{equation}
    \label{inequality2}
    \left| \left\langle \Phi\left| e^{itH_{\varepsilon}} H_I e^{-itH_{\varepsilon}}\right| \Phi\right\rangle_{\mathscr{H}}  \right| \leqslant A \left\langle \Phi\left| e^{itH_{\varepsilon}} H_+ e^{-itH_{\varepsilon}} \right|\Phi \right\rangle_{\mathscr{H}} + B \left\| \Phi \right\|_{\mathscr{H}}^2 \; .
  \end{equation}
  
  It remains to bound the first term of the sum. We will use a Gronwall and density argument. Without loss
  of generality, we can assume that $t>0$. Let $-M$, $M \geqslant 0 $, be a lower bound for $H_{\varepsilon}$, and suppose that $\Phi\in
  \mathscr{D}\bigl((H_{\varepsilon}+M)^{3/2}\bigr)$, the latter being a dense domain. Since $e^{-itH_{\varepsilon}}$ is weakly
  differentiable on $\mathscr{D}[H_0]\supset \mathscr{D}\bigl((H_{\varepsilon}+M)^{3/2}\bigr)$, and maps
  $\mathscr{D}\bigl((H_{\varepsilon}+M)^{3/2}\bigr)$ into itself for any $t\in \mathbb{R}$, it is possible to write
  \begin{multline*} 
      \left\langle \Phi\left| e^{itH_{\varepsilon}} (H_+  +H_I) e^{-itH_{\varepsilon}} \right|\Phi \right\rangle_{\mathscr{H}} \leqslant \left\langle \Phi\left|  H_+  +H_I \right| \Phi\right\rangle_{\mathscr{H}} \\
      + \left| \int_0^t \mathrm{d} s \: \partial_s  \left\langle \Phi \left| e^{isH_{\varepsilon}}(H_+  +H_I) e^{-isH_{\varepsilon}}\right| \Phi \right\rangle_{\mathscr{H}} \right| \\
      \leqslant \left\langle \Phi\left|  H_+  +H_I \right| \Phi\right\rangle_{\mathscr{H}} + \left| \int_0^t \mathrm{d} s \:   \left\langle \Phi \left| e^{isH_{\varepsilon}} \left[H_+  +H_I, H_0 + H_I \right] e^{-isH_{\varepsilon}}\right| \Phi \right\rangle_{\mathscr{H}} \right|.
   \end{multline*}
  Now, the commutator, as a quadratic form, yields
  \begin{equation*}
    \left[H_+  +H_I,H_0+H_I \right]= \left[H_+, H_I\right] + \left[H_I, H_0\right] = (\varepsilon-1) \lf(b^{\dagger}_{\varepsilon}\big(\lambda^{(b)}_{\mathbf{x}}\big) - b_{\varepsilon}\big(\lambda^{(b)}_{\mathbf{x}}\big) \ri)\; .
  \end{equation*}
  It is well-known \cite[Prop. A.2]{correggi2017ahp} that $b^{\dagger}_{\varepsilon}(\lambda^{(b)}_{\mathbf{x}}) -
  b_{\varepsilon}(\lambda^{(b)}_{\mathbf{x}})$ is small in the sense of quadratic forms w.r.t.\ $H_+$, with $A' \in (0,
  1/4)$ and $B' > 0$, so that, for any $\Theta\in \mathscr{D}[H_+]$,
  \begin{equation}\label{caextimate}
    \left| \left\langle \Theta\left| \left( b^{\dagger}_{\varepsilon}\big(\lambda^{(b)}_{\mathbf{x}}\big) - b_{\varepsilon}\big(\lambda^{(b)}_{\mathbf{x}}\big) \right) \Theta\right. \right\rangle_{\mathscr{H}} \right| \leqslant 2 \left| \left\langle \Theta\left|b^{\dagger}_{\varepsilon}\big(\lambda^{(b)}_{\mathbf{x}}\big) \Theta\right. \right\rangle_{\mathscr{H}} \right| \leqslant  2 A' \left\langle \Theta\left|H_+\right| \Theta\right\rangle_{\mathscr{H}} + 2B' \left\| \Theta \right\|^2_{\mathscr{H}}\; .
  \end{equation}
  By KLMN theorem, the above bound yields an estimate from below for $H_+ + H_I$: for any $\Theta\in
  \mathscr{D}[H_+]$,
  \begin{equation}
     \left\langle \Theta\left| H_+  + H_I \right|\Theta\right\rangle_{\mathscr{H}}  \geqslant (1-A') \left\langle \Theta\left|  H_+\right| \Theta\right\rangle_{\mathscr{H}} - B' \left\| \Theta \right\|^2_{\mathscr{H}} 
  \end{equation}
  or, equivalently,
  \begin{equation}\label{klmnextimate}
    A' \left\langle \Theta\left| H_+\right|\Theta\right\rangle_{\mathscr{H}}+ \left\langle \Theta\left| H_I \right| \Theta\right\rangle_{\mathscr{H}}  + B' \left\| \Theta \right\|^2_{\mathscr{H}} \geqslant 0\; .
  \end{equation}
  Now, we use \eqref{klmnextimate} in \eqref{caextimate} to obtain (remembering that $3A'<1$)
  \begin{multline*} 
   \left| \left\langle \Theta\left| \left( b^{\dagger}_{\varepsilon}\big(\lambda^{(b)}_{\mathbf{x}}\big) - b_{\varepsilon}\big(\lambda^{(b)}_{\mathbf{x}}\big) \right) \Theta\right. \right\rangle_{\mathscr{H}} \right| \leqslant   3A' \left\langle \Theta\left| H_+\right|\Theta\right\rangle_{\mathscr{H}}+ \left\langle \Theta\left| H_I \right| \Theta\right\rangle_{\mathscr{H}}  + 3B' \left\| \Theta \right\|^2_{\mathscr{H}} \\
   \leqslant \left\langle \Theta\left| H_+  + H_I \right|\Theta\right\rangle_{\mathscr{H}}  +3B' \left\|\Phi \right\|^2_{\mathscr{H}}\; .
   \end{multline*}
  Therefore, using the fact that $|1-\varepsilon| < 1$, we get
  \begin{equation*}
    \begin{split}
        \left\langle \Phi\left| e^{itH_{\varepsilon}} (H_+  +H_I) e^{-itH_{\varepsilon}} \right|\Phi \right\rangle_{\mathscr{H}} \leqslant  \int_0^t \mathrm{d} s \: \left\langle \Phi \left| e^{isH_{\varepsilon}}(H_+  +H_I) e^{-isH_{\varepsilon}}\right| \Phi \right\rangle_{\mathscr{H}} \\
        + \left\langle \Phi\left|  H_+  +H_I \right| \Phi\right\rangle_{\mathscr{H}} +3B'  t   \left\|\Phi \right\|^2_{\mathscr{H}}\; .
    \end{split}
  \end{equation*}
  If we now set
  \begin{equation*}
    P(t):=\left\langle \Phi \left| e^{itH_{\varepsilon}}(H_+  +H_I) e^{-itH_{\varepsilon}}\right| \Phi \right\rangle_{\mathscr{H}}\; ,	\qquad		 B(t) := 3B' t \left\|\Phi \right\|_{\mathscr{H}}^2\; ,
  \end{equation*}
  we can rewrite the previous inequality as (recall that $ t \geqslant 0 $)
  \begin{equation}
    P(t) \leqslant P(0) + \int_0^t\mathrm{d} \tau \:  P(\tau) + B(t)\; .
  \end{equation}
  Gronwall lemma then yields
  \begin{displaymath}
      P(t) \leqslant P(0) + B(t) + \int_0^t \mathrm{d} \tau \: \left(P(0) + B(\tau)\right)  e^{t-\tau} \leqslant e^t \left\langle \Phi\left|  H_+  +H_I \right| \Phi\right\rangle_{\mathscr{H}} + 3B't e^t \left\| \Phi \right\|^2_{\mathscr{H}} \; .
  \end{displaymath}
  Therefore inserting this bound  in \eqref{prop:firstineq}, we get for the expectation of $H_+$
  (considering now also the analogous case $t<0$):
  \begin{equation}
    \label{eq:6}
    (1-A)\left\langle \Phi \left| e^{itH_{\varepsilon}} H_+ e^{-itH_{\varepsilon}}\right| \Phi \right\rangle_{\mathscr{H}} \leqslant e^{\lvert t  \rvert_{}^{}} \left\langle \Phi\left|  H_+  +H_I \right| \Phi\right\rangle_{\mathscr{H}} + \left(3B'\lvert t  \rvert_{}^{}e^{\lvert t  \rvert_{}^{}}+B \right) \left\| \Phi \right\|^2_{\mathscr{H}}\; .
  \end{equation}
  This concludes the proof for $\Phi\in \mathscr{D}\bigl((H_{\varepsilon}+M)^{3/2}\bigr)$, since $A<1$. The proof is then
  extended to any $\Phi\in \mathscr{D}[H_0]$ by a density argument.
\end{proof}

	\begin{lemma}
  		\label{lemma:1}
  		\mbox{}	\\
  		Let $\mathcal{W}\in C_0^{\infty}(\mathbb{R}^3)$. Then, for any $t,s\in \mathbb{R}$, there exists a constant $ C_{t,s}>0$ depending only on $\lVert \nabla
  \mathcal{W} \rVert_{\infty}^{}$ and $\lVert \Delta \mathcal{W} \rVert_{\infty}^{}$, such that, for every $\Phi\in \mathscr{D}[H_0]$,
  		\begin{equation}
    			\left\langle \Phi\left|{Y}^{\dagger}_{\sigma}(t,s) (-\Delta\otimes 1) {Y}_{\sigma}(t,s)\right|\Phi\right\rangle_{\mathscr{H}} \leqslant C_{t,s} \sigma^{-6} \left( \left\| (\nabla\otimes 1) \Phi \right\|_{\mathscr{H}}^2+ \left\| \Phi \right\|_{\mathscr{H}}^2 \right)\; .
  		\end{equation}
	\end{lemma}

	\begin{proof}
  First of all, let us notice that ${Y}_{\sigma}(t,s)^{\dagger}(\Delta\otimes 1) {Y}_{\sigma}(t,s)=\mathcal{U}_{\sigma}^{\dagger}(t,s)\Delta\,
  \mathcal{U}_{\sigma}(t,s)\otimes 1$. Omitting the multiplication by the identity when it is obvious from the context, we can write
  \begin{equation*}
    \left\langle \Phi\left|{Y}^{\dagger}_{\sigma}(t,s) (-\Delta\otimes 1) {Y}_{\sigma}(t,s)\right|\Phi\right\rangle_{\mathscr{H}} = \left\langle \mathcal{U}_{\sigma}(t,s) \Phi\left|-\Delta\right|\mathcal{U}_{\sigma}(t,s) \Phi\right\rangle_{\mathscr{H}} =: R(t,s) \; .
  \end{equation*}
  Suppose now that $\Phi\in \mathscr{D}(H_0)$. Without loss of generality, we can also assume that $t\geqslant s\geqslant
  0$. Hence, we get the integral inequality
  \begin{displaymath}
      R(t,s)\leqslant R(s,s)+\int_s^t \mathrm{d} \tau \: \left| \left\langle \mathcal{U}_{\sigma}(t,s) \Phi|\left[-\Delta, \mathcal{W}_{{\beta(\tau),\sigma}}(\mathbf{x}) \right] \mathcal{U}_{\sigma}(t,s) \Phi \right\rangle_{\mathscr{H}} \right| \;, 
  \end{displaymath}
  where $ \nu(\sigma) $ and $ \mathcal{W}_{\sigma} $ are given by \eqref{eq:nu} and \eqref{eq:vsigma}, respectively.
  The commutator has the following explicit form:
  \begin{equation*}
    [-\Delta, \mathcal{W}_{{\beta(\tau),\sigma}}]= \left( -\Delta \mathcal{W}_{{\beta(\tau),\sigma}} \right)(\mathbf{x}) + 2 \left(\nabla \mathcal{W}_{{\beta(\tau),\sigma}} \right)(\mathbf{x}) \cdot \nabla\; .
  \end{equation*}
  Therefore, keeping in mind that $\lvert \beta(t) \rvert_{}^{}\leqslant C $ for all $t\in \mathbb{R}$, we obtain
  \begin{equation*}
    \begin{split}
      R(t,s)\leqslant R(s,s)+(t-s)\sigma^{-4} \lVert \Delta \mathcal{W}  \rVert_{\infty}^{} \left\| \Phi  \right\|_{\mathscr{H}}^2 + 2\sigma^{-3} \lVert \nabla \mathcal{W}  \rVert_{\infty}^{} \left\| \Phi  \right\|_{\mathscr{H}}^{}\int_s^t \mathrm{d}\tau \: \left\| \nabla\, \mathcal{U}_{\sigma}(\tau,s) \Phi \right\|_{\mathscr{H}}^{}\\
      \leqslant R(s,s)+ \left[ (t-s)\sigma^{-4} \left\| \Delta \mathcal{W}  \right\|_{\infty}^{} +\sigma^{-6} \left\| \nabla \mathcal{W}  \right\|_{\infty}^2 \right] \left\| \Phi  \right\|_{\mathscr{H}}^2+\int_s^t  \mathrm{d}\tau \: R(\tau,s)\; .
    \end{split}
  \end{equation*}
  The result for $\Phi\in \mathscr{D}(H_0)$ then follows by Gronwall lemma, and it is extended by density to
  $\Phi\in \mathscr{D}[H_0]$.
\end{proof}

We can now conclude the proof of \cref{pro:1}.

\begin{proof}[Proof of \cref{pro:1}.]
 Since the operators $ Z_{\varepsilon} $ and $ Z $ are bounded, in order to prove strong convergence it suffices to show that $ Z_{\varepsilon} $ weakly converges to $ Z $. Therefore, we want to prove that for
  all $ \Phi,\Theta\in \mathscr{H}$,
  \begin{equation*}
    \lim_{\varepsilon\to 0} \, \left\langle \Theta\left|  \bigl(Z(t,s)-Z_{\varepsilon}(t,s)\bigr) \Phi \right. \right\rangle_{\mathscr{H}} =0 \; ;
  \end{equation*}
  however, by triangular inequality, we can show separately the convergences of $ Z_{\varepsilon} $ to $ {Y}_{\sigma} $ and of $ {Y}_{\sigma} $ to $ Z  $. By \cref{pro:approx}, we already know that $\mathcal{U}_{\sigma}$ converges strongly to $\mathcal{U}_{\mathrm{eff}}$ on $L^2 (\mathbb{R}^3)$, which implies that $ {Y}_{\sigma} $ converges to $ Z $, provided that $ \sigma \to 0 $, as $ \varepsilon \to 0 $, which we are going to assume. Therefore, it suffices to prove convergence of $ Z_\varepsilon $ to $ {Y}_{\sigma} $.
  
  We restrict to the dense set
  \begin{equation}
  	\label{eq:denseset}
    \mathscr{D}:=\mathscr{D}(H_0)\cap \bigl(\mathscr{D}\bigl(\mathrm{d} \bm{\Gamma}(((k^2+1)^{2M},(k^2+1)^{2N}))^{1/2}\bigr)\subset\mathscr{D}[H_0]\; ,
  \end{equation}
  where $M$ is defined by \eqref{eq:11}, and $N> 1$, and prove weak convergence to zero on $\mathscr{D}$
  for the quadratic form associated to $ {Y}^{\dagger}_{\sigma}(t,s)Z_{\varepsilon}(t,s)-1 $, whenever $\sigma=\mathcal{O}(\varepsilon^{\gamma})$,
  with suitable $\gamma>0$. Note that, since $ k^2+1\geqslant 1 $, $ \mathrm{d} \bm{\Gamma}(\mathbf{1}) \leqslant \mathrm{d}
  \bm{\Gamma}(((k^2+1)^{2M},(k^2+1)^{2N})) $ and consequently $ \mathscr{D} \bigl(\mathrm{d} \bm{\Gamma}(((k^2+1)^{2M},(k^2+1)^{2N}))^{1/2}\bigr) \subset
  \mathscr{D} \bigl(\mathrm{d} \bm{\Gamma}(\mathbf{1})^{1/2}\bigr) $. By the polarization identity
  \begin{multline*} 
    	\left\langle \Theta\left|  \bigl({Y}_{\sigma}-Z_{\varepsilon} \bigr) \Phi \right. \right\rangle_{\mathscr{H}} = \left\langle {Y}_{\sigma}^{\dagger} \Theta\left| \bigl( {Y}_{\sigma}^{\dagger}Z_{\varepsilon} - 1 \bigr) \Phi \right. \right\rangle_{\mathscr{H}} \\
    	= \sum_{j=1}^4 \zeta_j \left\langle \eta_j \left({Y}_{\sigma}^{\dagger} \Theta, \Phi \right)\left| \bigl( {Y}_{\sigma}^{\dagger}Z_{\varepsilon} - 1 \bigr) \eta_j \left({Y}_{\sigma}^{\dagger} \Theta, \Phi \right) \right. \right\rangle_{\mathscr{H}}\; ,
 	 \end{multline*}
  where the coefficients $\zeta_j\in \mathbb{C}$ and the linear combinations $\eta({Y}_{\sigma}^{\dagger}\Theta,\Phi)_j\in \mathscr{H}$ of $Z^{\dagger}_{\sigma}\Theta$ and
  $\Phi$ are suitably given. Hence, on one hand, weak convergence of
  ${Y}_{\sigma}(t,s)-Z_{\varepsilon}(t,s)$ implies convergence of the quadratic form, via the identification $\Theta={Y}_{\sigma}\Psi$ {and \cref{pro:approx}}, and, on the
  other hand, if
  \begin{equation}
    \label{eq:5}
    \left\langle \Theta\left| \bigl( {Y}_{\sigma}^{\dagger}Z_{\varepsilon} - 1 \bigr) \Phi \right. \right\rangle_{\mathscr{H}} \xrightarrow[\varepsilon \to 0]{} 0,		\qquad		\forall \Phi\in \mathscr{H},
  \end{equation}
  then, by the same polarization identity, it follows that $ {Y}_{\sigma}(t,s)-Z_{\varepsilon}(t,s) $ converges
  weakly to zero. 
  
  Furthermore, by the uniform boundedness in $ \varepsilon $ of both ${Y}_{\sigma}$
  and $Z_{\varepsilon}$, it is sufficient to prove the convergence of the quadratic form on the dense
  set $\mathscr{D}$, which we are going to do now. The operator $Z_{\sigma}^{\dagger}$ is strongly differentiable on $\mathscr{D}(H_0)$ and maps $\mathscr{D}[H_0]$
  into itself, while $Z_{\varepsilon}$ is weakly differentiable on $\mathscr{D}[H_0]$. Without loss of generality,
  we can suppose that $t\geqslant s\geqslant 0$. Therefore, it is possible to write
  \begin{multline*} 
       \left| \left\langle \Phi\left|\left({Y}_{\sigma}^{\dagger}Z_{\varepsilon} -1 \right)\Phi\right. \right\rangle_{\mathscr{H}} \right| \leqslant \int_s^t  \mathrm{d}\tau \: \left| \partial_\tau \left\langle \Phi\left|\left({Y}_{\sigma}^{\dagger}(\tau,s) Z_{\varepsilon}(\tau,s) -1 \right)\Phi\right. \right\rangle_{\mathscr{H}} \right| \\
      =  \int_s^t  \mathrm{d}\tau \: \left|  \left\langle \Phi\left| {Y}_{\sigma}^{\dagger}(\tau,s) \left( L_{\varepsilon}(\tau) - {\mathcal{K}_{\beta(\tau),\sigma} \otimes 1 \otimes 1 - 1 \otimes \mathrm{d}\bm{\Gamma}((0,\kappa))} \right) Z_{\varepsilon}(\tau,s) \Phi\right. \right\rangle_{\mathscr{H}} \right| \;.
    \end{multline*}
  Now, since by hypothesis $2\Re \left\langle \bm{\lambda}_{\mathbf{x}}\left| \bm{\alpha}_{\varepsilon}(t)\right. \right\rangle_{\mathfrak{H}\oplus \mathfrak{H}}  = \mathcal{W}_{{\beta(t),\sigma}}(\mathbf{x}) $,
  and $\frac{1}{\varepsilon}\mathrm{d}\Gamma^{(b)}_{\varepsilon}(1)=\mathrm{d}\Gamma^{(b)}(1)$ by \eqref{eq:dgamma}, we have that
  \begin{equation*}
    L_{\varepsilon}(t) - {\mathcal{K}_{\beta(\tau),\sigma} \otimes 1 \otimes 1 - 1 \otimes \mathrm{d}\bm{\Gamma}((0,\kappa))}  = \mathbf{a}_{\varepsilon}(\bm{\lambda}_{\mathbf{x}})+ \mathbf{a}_{\varepsilon}^{\dagger}(\bm{\lambda}_{\mathbf{x}})=\sqrt{\varepsilon} \lf(\mathbf{a}(\bm{\lambda}_{\mathbf{x}})+\mathbf{a}^{\dagger}(\bm{\lambda}_{\mathbf{x}}) \ri)\; .
  \end{equation*}
  Therefore, we obtain the following bound
  \begin{equation}
    \label{sqrtepscontrol}
     \left| \left\langle \Phi\left|\left({Y}_{\sigma}^{\dagger}Z_{\varepsilon} -1 \right)\Phi\right. \right\rangle_{\mathscr{H}} \right| \leqslant \sqrt{\varepsilon}\int_s^t\mathrm{d}\tau \: \left| \left.\left\langle \left(\mathbf{a}(\bm{\lambda}_{\mathbf{x}})+\mathbf{a}^{\dagger}(\bm{\lambda}_{\mathbf{x}}) \right) {Y}_{\sigma}(\tau,s)\Phi\right| Z_{\varepsilon}(\tau,s)\Phi \right\rangle_{\mathscr{H}}  \right|   \; .
  \end{equation}
  
  Let us define
  \begin{align*}
    S &:= \int_s^t\mathrm{d}\tau \: \left| \left.\left\langle  \mathbf{a}(\bm{\lambda}_{\mathbf{x}})  {Y}_{\sigma}(\tau,s)\Phi\right| Z_{\varepsilon}(\tau,s)\Phi \right\rangle_{\mathscr{H}}  \right| \; , \\
    T &:= \int_s^t\mathrm{d}\tau \: \left| \left.\left\langle  \mathbf{a}^{\dagger}(\bm{\lambda}_{\mathbf{x}})  {Y}_{\sigma}(\tau,s)\Phi\right| Z_{\varepsilon}(\tau,s)\Phi \right\rangle_{\mathscr{H}}  \right| \; .
  \end{align*}
  The term $S$ is easy to bound, exploiting the fact that
    $\mathrm{d}\bm{\Gamma}(((k^2+1)^{M},(k^2+1)^{N}))^{1/2}$ commutes with $ \mathrm{d} \bm{\Gamma}((0,\kappa)) $, and
    therefore with $Z_{\sigma}$:
  \begin{multline} 
	\label{eq: S}
    S \leqslant \int_s^t\mathrm{d}\tau \: \left\| \mathbf{a}(\bm{\lambda}_{\mathbf{x}}) {Y}_{\sigma}(\tau,s) \Phi \right\|_{\mathscr{H}} \left\| \Phi \right\|_{\mathscr{H}} \leqslant (t-s)  \sup_{\tau \in (s,t)} \left\| \mathbf{a}(\bm{\lambda}_{\mathbf{x}}) {Y}_{\sigma}(\tau,s) \Phi \right\|_{\mathscr{H}} \left\| \Phi \right\|_{\mathscr{H}} \\
    \leqslant  (t-s)   \left\| \lf((k^2+1)^{-M},(k^2+1)^{-N}\ri) \lambda_{\mathbf{x}} \right\|_{L^{\infty}(\mathbb{R}^3; \mathfrak{H}\oplus \mathfrak{H})} \left\| \mathrm{d} \bm{\Gamma}\lf((k^2+1)^{2M},(k^2+1)^{2N}\ri)^{1/2} \Phi \right\|_{\mathscr{H}} \left\| \Phi \right\|_{\mathscr{H}} \\\leqslant C\left\| \mathrm{d} \bm{\Gamma}\lf((k^2+1)^{2M},(k^2+1)^{2N}\ri)^{1/2} \Phi \right\|_{\mathscr{H}} \left\| \Phi \right\|_{\mathscr{H}},
  \end{multline}
 which is uniformly bounded on \eqref{eq:denseset}.
 
  The $T$ term requires some additional care: recalling \eqref{lambdacond}, we observe that one can split $T$ into a regular and a singular part, {\it i.e.}, $ T \leqslant T_{\text{reg}} + T_{\text{sing}}  $, where
  \begin{align*}
    T_{\text{reg}} &:= \int_s^t\mathrm{d}\tau \: \left| \left.\left\langle  \mathbf{a}^{\dagger}(\bm{\lambda}_{<,\mathbf{x}})  {Y}_{\sigma}(\tau,s)\Phi\right| Z_{\varepsilon}(\tau,s)\Phi \right\rangle_{\mathscr{H}}  \right| \; ,\\
    T_{\text{sing}} &:= \int_s^t\mathrm{d}\tau \: \left| \left.\left\langle  \mathbf{a}^{\dagger}(\bm{\lambda}_{>,\mathbf{x}})  {Y}_{\sigma}(\tau,s)\Phi\right| Z_{\varepsilon}(\tau,s)\Phi \right\rangle_{\mathscr{H}}  \right| \; .
  \end{align*}
  The regular part can be treated analogously to $S$: let $\bm{\omega}$ be the operator on $\mathfrak{H}\oplus
  \mathfrak{H}$ acting as $\bm{\omega}\bm{\eta}=(\omega\eta_1,\kappa\eta_2)$, then
  \begin{equation}
  	\label{eq: Treg}
    T_{\text{reg}} \leqslant (t-s) \lf(\| \bm{\omega}^{-1/2}\bm{\lambda}_{<,\mathbf{x}} \|_{L^{\infty}(\mathbb{R}^3; \mathfrak{H}\oplus \mathfrak{H})}+\left\| \bm{\lambda}_{<,\mathbf{x}} \right\|_{L^{\infty}(\mathbb{R}^3; \mathfrak{H}\oplus \mathfrak{H})}\ri) \left\| (\mathrm{d} \bm{\Gamma}(\bm{\omega})+1)^{1/2}  \Phi \right\|_{\mathscr{H}}  \left\|\Phi \right\|_{\mathscr{H}}\; ,
  \end{equation}
  which is uniformly bounded w.r.t.\ $\varepsilon$. 
  
  It remains to estimate the singular part. Let us split again such term in
  two parts: $ T_{\text{sing}} = \widetilde{T}_1 + \widetilde{T}_2$, where
  \begin{align*}
    \widetilde{T}_1 & := \int_s^t\mathrm{d}\tau \: \left| \left.\left\langle  \mathbf{a}^{\dagger}(\bm{\xi}_{\mathbf{x}}) \cdot \nabla {Y}_{\sigma}(\tau,s)\Phi\right| Z_{\varepsilon}(\tau,s)\Phi \right\rangle_{\mathscr{H}}  \right| \; ,\\
    \widetilde{T}_2 & := \int_s^t\mathrm{d}\tau \: \left| \left.\left\langle \nabla \cdot  \mathbf{a}^{\dagger}(\bm{\xi}_{\mathbf{x}})  {Y}_{\sigma}(\tau,s)\Phi\right| Z_{\varepsilon}(\tau,s)\Phi \right\rangle_{\mathscr{H}}  \right| \; .
  \end{align*}
  The first one is bounded as follows:
  \begin{multline*}  
  \widetilde{T}_1 \leqslant (t-s) \sup_{\tau \in (s,t)} \left\| a^{\dagger}(\bm{\xi}_{\mathbf{x}}) \cdot \nabla {Y}_{\sigma}(\tau,s) \Phi \right\|_{\mathscr{H}} \left\|
    \Phi \right\|_{\mathscr{H}}	
    \leqslant C \sup_{\tau \in (s,t)} \sum_{j=1}^3 \left\| a^{\dagger}\left(\left(\bm{\xi}_{\mathbf{x}}\right)_j \right) \partial_{j} {Y}_{\sigma}(\tau,s) \Phi \right\| \left\|
    \Phi \right\|_{\mathscr{H}} \\ \leqslant C \Bigl(\bigl\| \bm{\omega}^{-1/2}\bm{\xi}_{\mathbf{x}} \bigr\|_{L^{\infty}(\mathbb{R}^3; (\mathfrak{H}\oplus \mathfrak{H})^3)} \left\| \nabla_{\mathbf{x}} {Y}_{\sigma}(\tau,s) \mathrm{d} \Gamma(\bm{\omega})^{1/2} \Phi \right\|_{\mathscr{H}} \\+ \bigl\| \bm{\xi}_{\mathbf{x}} \bigr\|_{L^{\infty}(\mathbb{R}^3; (\mathfrak{H}\oplus \mathfrak{H})^3)} \left\| \nabla_{\mathbf{x}} {Y}_{\sigma}(\tau,s)  \Phi \right\|_{\mathscr{H}}\Bigr)\left\|
    \Phi \right\|_{\mathscr{H}} 
      \\\leqslant  C \Bigl(\left| \left\langle {Y}_{\sigma}(\tau,s)\,\Theta \left| -\Delta\right|{Y}_{\sigma}(\tau,s)\,\Theta \right\rangle_{\mathscr{H}} \right|^{1/2}+\left| \left\langle {Y}_{\sigma}(\tau,s)\,\Phi \left| -\Delta\right|{Y}_{\sigma}(\tau,s)\,\Phi \right\rangle_{\mathscr{H}} \right|^{1/2}\Bigr),
	 \end{multline*}
  with $\Theta := \mathrm{d} \Gamma(\bm{\omega})^{1/2} \Phi$. Therefore, by Lemma~\ref{lemma:1}, it follows that there
  exists a constant $C_{t,s}>0$, depending on $\Phi$, such that
  \begin{equation}
  	\label{eq: T1}
    \widetilde{T}_1\leqslant C_{t,s} \sigma^{-3}.
  \end{equation}
  
  The second term is bounded using Lemma~\ref{laplacecontrol}:
  \begin{multline*} 
      \widetilde{T}_2 \leqslant \sum_{j=1}^3 \left\| a^{\dagger}\left(\left(\bm{\xi}_{\mathbf{x}}\right)_j \right) {Y}_{\sigma}(\tau,s) \Phi \right\|_{\mathscr{H}}  \left\|\partial_j Z_{\varepsilon}(\tau,s) \Phi \right\|_{\mathscr{H}} \leqslant C \Bigl(\bigl\| \bm{\omega}^{-1/2}\bm{\xi}_{\mathbf{x}} \bigr\|_{L^{\infty}(\mathbb{R}^3; (\mathfrak{H}\oplus \mathfrak{H})^3)} \left\| \mathrm{d} \Gamma(\bm{\omega})^{1/2} \Phi \right\|_{\mathscr{H}}\\+\left\| \bm{\xi}_{\mathbf{x}} \right\|_{L^{\infty}(\mathbb{R}^3; (\mathfrak{H}\oplus \mathfrak{H})^3)} \left\|  \Phi \right\|_{\mathscr{H}}\Bigr)\left\| (-\Delta)^{1/2} Z_{\varepsilon}(\tau,s) \Phi \right\|^{}_{\mathscr{H}}.
    \end{multline*}
  Now, setting $\mathbf{W}_{t}:= \mathbf{W}_{\varepsilon}\bigl( \frac{\bm{\alpha}_{\varepsilon}(t)}{i\varepsilon}\bigr)$  and using the fact that $-\Delta$ commutes with
  $\mathbf{W}_t$ for any $t\in \mathbb{R}$, we get by Lemma~\ref{laplacecontrol}
  \begin{displaymath}
      \left\| (-\Delta)^{1/2} Z_{\varepsilon}(\tau,s) \Phi \right\|_{\mathscr{H}}^2 \leqslant C_t \left( \left\langle \mathbf{W}_s \Phi\left|H_+  + H_I\right|\mathbf{W}_s \Phi\right\rangle_{\mathscr{H}} + \left\| \Phi \right\|^2_{\mathscr{H}} \right)
  \end{displaymath}
  However, by the translation properties of Weyl operators, one has
  \begin{equation*}
    \mathbf{W}^{\dagger}_s (H_+  + H_I) \mathbf{W}_s = -\Delta  + \mathrm{d} \bm{\Gamma}_{\varepsilon}(\bm{\omega}) + H_I + \left\|\bm{\omega}\bm{\alpha}_{\varepsilon}(s) \right\|^2 + \mathbf{a}\bigl(\bm{\omega}\bm{\alpha}_{\varepsilon}(s)\bigr) + \mathbf{a}^{\dagger}\bigl(\bm{\omega}\bm{\alpha}_{\varepsilon}(s)\bigr) + 2 \Re \left\langle \bm{\alpha}_{\varepsilon}(s)\left|\bm{\lambda}_{\mathbf{x}}\right. \right\rangle\;,
  \end{equation*}
  so that
  \begin{multline*} 
      \left\| (-\Delta)^{1/2} Z_{\varepsilon}(\tau,s) \Phi \right\|_{\mathscr{H}}^2 \leqslant C_t \left[ \left\langle  \Phi \left|-\Delta + \mathrm{d} \bm{\Gamma}_{\varepsilon}(\bm{\omega}) + H_I +\mathbf{a}\bigl(\bm{\omega}\bm{\alpha}_{\varepsilon}(s)\bigr) + \mathbf{a}^{\dagger}\bigl(\bm{\omega}\bm{\alpha}_{\varepsilon}(s)\bigr) \right|\Phi \right\rangle_{\mathscr{H}} \right.\\ + \left. \left( 2 \Re \left\langle \bm{\alpha}_{\varepsilon}(s)\left|\bm{\lambda}_{\mathbf{x}}\right. \right\rangle + \left\| \bm{\omega}\bm{\alpha}_{\varepsilon}(s) \right\|_{\mathfrak{H}}^2 + 1 \right) \left\| \Phi \right\|_{\mathscr{H}}^2 \right]\; .
     \end{multline*}
  Now, since by \eqref{eq:11}
  	\beqn
  		\left\| \left\langle \bm{\alpha}_{\varepsilon}(s)\left|\bm{\lambda}_{\mathbf{x}}\right. \right\rangle_{\mathfrak{H}\oplus \mathfrak{H}} \right\|_{\infty} & \leqslant & {2 \lf\| \mathcal{W}_{\beta(t),\sigma}(\mathbf{x}) \ri\|_{\infty}} = \mathcal{O}\lf(\sigma_{\varepsilon}^{-2}\ri)\;,	\\
		\left\|\bm{\omega}\bm{\alpha}_{\varepsilon}(s) \right\|^2 & = &\mathcal{O}\lf(\sigma_{\varepsilon}^{-(1+{2+4M})}+\sigma_{\varepsilon}^{-1} \ri)\;,
	\eeqn
	{and the other terms are either uniformly bounded or subdominant with
        respect to $\sigma$.} It follows that
  \begin{equation}
  	\label{eq: T2}
    \widetilde{T}_2 \leqslant C \lf( \sigma^{{{-3-4M}}} + \sigma^{-1} \ri)\; .
  \end{equation}
  
Hence, putting together \eqref{eq: S} with \eqref{eq: Treg}, \eqref{eq: T1} and \eqref{eq: T2}, we finally get, for any $ \Phi \in \mathscr{D} $,
	\begin{displaymath}
		 \left| \left\langle \Phi\left|\left({Y}_{\sigma}^{\dagger}Z_{\varepsilon} -1 \right)\Phi\right. \right\rangle_{\mathscr{H}} \right| \leqslant \mathcal{O}\lf(\varepsilon^{1/2} \sigma^{-3}\ri) + \mathcal{O}\lf(\varepsilon^{1/2}\sigma^{{-3-4M}}\ri) = o(1),
	\end{displaymath}
	as long as $ \sigma \gg \varepsilon^{1/{j_*}} ${, with $ j_* = 6 + 8M $}.
\end{proof}

\end{document}